\documentclass[a4paper,11pt]{article}
\usepackage{typearea}\typearea{12}
\usepackage{amsmath,amsfonts,amssymb}
\usepackage{graphicx}
\usepackage{hyperref}
\hypersetup{hidelinks,
 pdftitle={On the uniqueness of translation-invariant ground states of the S=1 Heisenberg and J-D antiferromagnetic chains},
 pdfauthor={Hal Tasaki}}

\newtheorem{theorem}{Theorem}[section]
\newtheorem{lemma}[theorem]{Lemma}
\newtheorem{proposition}[theorem]{Proposition}

\makeatletter
\@addtoreset{equation}{section}
\makeatother

\def\eq#1\en{\begin{equation}#1\end{equation}}
\def\eqa#1\ena{\begin{align}#1\end{align}}
\def\eqg#1\eng{\begin{gather}#1\end{gather}}
\newcommand{\lb}[1]{\label{e:#1}}
\newcommand{\rlb}[1]{\eqref{e:#1}}
\newcommand{\nl}{\notag\\}
\newcommand{\qedm}{\ensuremath{\Box}}
\newenvironment{proof}{\par\medskip\noindent{\em Proof.}\ }{\hfill\qedm\par\medskip}

\newcommand{\abs}[1]{\left|#1\right|}
\newcommand{\norm}[1]{\left\Vert#1\right\Vert}

\newcommand{\bbC}{\mathbb C}
\newcommand{\bbR}{\mathbb R}
\newcommand{\bbZ}{\mathbb Z}

\newcommand{\calA}{\mathfrak A}
\newcommand{\Aloc}{\calA_{\rm loc}}
\newcommand{\calK}{\mathcal K}
\newcommand{\Tr}{\operatorname{Tr}}
\newcommand{\tr}{\operatorname{tr}}
\newcommand{\per}{\mathrm{per}}
\newcommand{\opp}{\mathrm{opp}}
\newcommand{\fld}{\mathrm{fld}}
\newcommand{\free}{\mathrm{free}}
\newcommand{\spn}{\boldsymbol S}
\newcommand{\Sx}{S^{\mathrm x}}
\newcommand{\Sy}{S^{\mathrm y}}
\newcommand{\Sz}{S^{\mathrm z}}
\newcommand{\one}{\mathbf 1}
\newcommand{\Ux}{U_{\mathrm x}}
\makeatletter
\long\def\@makecaption#1#2{{\small
\advance\leftskip1cm
\advance\rightskip1cm
\vskip\abovecaptionskip
\sbox\@tempboxa{#1: #2}%
\ifdim \wd\@tempboxa >\hsize
 #1: #2\par
\else
 \global\@minipagefalse
 \hb@xt@\hsize{\hfil\box\@tempboxa\hfil}%
\fi
\vskip\belowcaptionskip}}
\makeatother

\begin{document}
\noindent
{\Large\bfseries\raggedright On the uniqueness of translation-invariant ground states of the $S=1$ Heisenberg and $J$--$D$ antiferromagnetic chains\par}

\renewcommand{\thefootnote}{\fnsymbol{footnote}}
\medskip\noindent
Hal Tasaki\footnote{Department of Physics, Gakushuin University,
Mejiro, Toshima-ku, Tokyo 171-8588, Japan.
E-mail: \texttt{hal.tasaki@gmail.com}.}
\renewcommand{\thefootnote}{\arabic{footnote}}
\setcounter{footnote}{0}

\begin{quotation}
\small\noindent
Using the analytic machinery and numerical inputs of OpenAI's proofs
of the Haldane gap, we prove uniqueness of the translation-invariant
ground state of the infinite $S=1$ Heisenberg antiferromagnetic chain
and of the $J$-$D$ model for sufficiently small $|D|$. Uniqueness is
among all translation-invariant ground states in the operator-algebraic
sense, without assuming a gap or spin-rotation symmetry for a competing
state. The proof combines concentration of the spatial transfer
spectrum, an exponentially small low-temperature free-energy correction,
and a Golden--Thompson comparison for separated local observables.
The unique translation-invariant state is pure, locally unique and
gapped, and continuous in $D$. The ground states of even periodic
chains and odd open chains with equal positive endpoint fields converge
to this state along their full sequences, uniformly in $D$, with a
stretched-exponential local rate. Every choice of ground states of
open chains without endpoint magnetic fields has the same local limit. These conclusions do not
establish uniqueness among arbitrary nontranslation-invariant ground
states of the infinite chain.

The author formulated the problems and proposed the core strategies.
ChatGPT Astra made substantial contributions to developing the proofs
and preparing the manuscript.
\end{quotation}

{\small\tableofcontents}

\section{Introduction}\label{S:introduction}

On October 6, 2026, OpenAI released a collection of mathematical
manuscripts that included two papers on the Haldane gap
\cite{OpenAIperiodic,OpenAIboundary}; see also the account in \cite{Tasaki_Haldane}.
These papers establish uniform gaps above unique
ground states for the $S=1$ antiferromagnetic Heisenberg chain on
sufficiently large even periodic rings and on sufficiently large odd
open chains with a fixed magnetic field at each endpoint.
The periodic paper \cite{OpenAIperiodic} settles the standard even-periodic
finite-chain formulation of the conjecture posed by Haldane more than
forty years earlier \cite{Haldane1981,Haldane1983a,Haldane1983b}.
See, e.g., \cite{TasakiBook} for the history and background.
The open-chain paper \cite{OpenAIboundary} establishes the boundary-field
version of the gap hypothesis formulated in \cite{Tasaki2025} to connect
a spectral gap to nontrivial topological properties of the ground state.
Its implications for the operator-algebraic index of the $S=1$
antiferromagnetic Heisenberg chain have already been discussed in
\cite{TasakiEquivalence}.

In studying the proofs in \cite{OpenAIperiodic,OpenAIboundary}, we
came to realize that the underlying methods may be useful for
investigating aspects of quantum many-body systems quite different
from the existence of a spectral gap. The particularly useful ingredients are a
self-adjoint spatial transfer operator, the probability distribution
formed from its even spectral moments, and the coupled improvement of
spatial and temporal purities. The temporal distribution describes the
physical Gibbs state, whereas the spatial distribution describes
transfer in the auxiliary space of bond histories. Concentration in
both directions contains information beyond the existence of a
finite-chain energy gap.

In this paper we extract such information to prove uniqueness among
all translation-invariant ground states of the infinite chain.
A competing translation-invariant ground state need not be pure,
gapped, or invariant under spin rotations. The conclusion persists
under a sufficiently small single-ion anisotropy.
We then show that the
periodic and boundary-field constructions converge to this state along
their full sequences, and that the state is continuous in the
anisotropy parameter. The boundary-field conclusion follows by proving
that every subsequential limit is translation invariant and applying
the translation-invariant uniqueness theorem. We also obtain the same local limit for
arbitrary choices of ground states of open chains without endpoint
magnetic fields.

The presentation in this paper follows the language and notation of our
continuous-time expository note \cite{TasakiCT}, referred to below as
the \emph{CT note}. We recall the transfer construction and the
periodic purity argument needed here, so the analytic proof does not
presuppose familiarity with that note.
The numerical thermal and variational certificates of
\cite{OpenAIperiodic,OpenAIboundary} remain imported inputs.

\medskip
A key observation for the infinite-chain limit is that, at the low
temperatures used below and sufficiently far from the endpoints,
local Gibbs expectations are governed by the dominant eigenvector
of the spatial transfer operator $K_{\beta,D}$.
Finite-chain Hamiltonians depend on the boundary conditions, whereas
this operator for the interior of the chain does not. Domination by
its leading eigenvector therefore makes local Gibbs expectations
asymptotically independent of the endpoints and translation
invariant. Quantitative estimates at suitable joint length and
temperature scales allow these properties to pass to ground states
of chains with endpoint fields. For open chains without endpoint
fields, the Gibbs purity need only remain bounded below by a positive
constant. Purity of the infinite-chain state then identifies the
limit of every choice of ground vectors, without requiring the
finite-chain Gibbs state to approach a rank-one projection.

The main new observation for the uniqueness proof is that the same
concentration estimates
control more than one local insertion. By placing copies of a local
observable in widely separated blocks, we bound the logarithm of their
exponential moment per block. The comparison of leading transfer
eigenvalues at consecutive temperatures also shows that the
free-energy density approaches the ground-state energy density
with an exponentially small error.
The minimum-energy characterization of translation-invariant ground
states \cite{BKR,BR2}, in an elementary form proved below, and the
Golden--Thompson inequality then force every such state to have the
same local expectations. Importantly, this argument does not assume
that the fixed-temperature comparison states already converge as the
temperature tends to zero.

\medskip 
There is a useful precedent in the AKLT paper
\cite[Section~2.4, Theorem~2.7 and the following remark]{AKLT1988}.
There, uniqueness among translation-invariant ground states in the
operator-algebraic sense follows by combining the minimum-energy
characterization with the exact structure of the finite-chain VBS
ground-state spaces. Frustration-freeness forces every restriction
of a competing state to be supported on the corresponding
zero-energy space, and the exact overlap formulas determine its
local expectations. Here we use the same variational starting point,
but replace those exact local constraints by a sufficiently sharp
low-temperature free-energy estimate and concentration of the
spatial transfer operator.

Our argument also has structural parallels with two works of Ogata.
Her large-deviation analysis of quantum spin chains
\cite{OgataLD} uses transfer operators to study logarithmic
exponential moments of sums of local observables. A variational
comparison between macroscopic expectations and exponential traces
appears in her work on approximating macroscopic observables by
commuting matrices \cite[Lemma~2.3]{OgataCommuting}.
We use an elementary form of this comparison together with the
low-temperature estimates to identify an arbitrary
translation-invariant ground state. We do not require a large-deviation
principle or a commuting approximation.

\medskip

We first derive the transfer representations of local expectations,
then prove the main theorem using only periodic information. The
boundary arguments follow separately. All analytical details are in
the main text; the appendix verifies the finite numerical criteria.
Applications to topological indices and phase transitions are left
to a forthcoming revision of the companion paper
\cite{TasakiEquivalence}.

\medskip\noindent
{\em Use of AI.}\/ 
After studying the analytic gap proofs and publishing the expository
notes listed in \cite{Tasaki_Haldane}, I realized that the same
machinery might control the infinite-chain limit of ground-state
expectation values. I asked ChatGPT Astra (OpenAI) to compare the
two finite-chain constructions and suggested investigating
convergence along the full sequences. ChatGPT developed proofs of
these statements and also proved continuity in $D$. After substantial
revision, this led to the first version of the present work
\cite{Tasaki_JD_v1}.
Soon after completing that version, I asked whether the OpenAI
machinery could also treat open chains without endpoint fields,
for which one does not expect a unique ground state separated from
the rest of the spectrum by a gap uniform in the chain length.
After obtaining an unexpected affirmative answer, corresponding
essentially to Section~\ref{S:free}, I asked ChatGPT to investigate
uniqueness by combining domination by a single transfer eigenvector
with general results available in \cite{BR1,BR2}.
ChatGPT then found the separated-insertion argument and its
connection with the minimum-energy principle, presented in
Section~\ref{S:uniqueness}, and completed the uniqueness proof.
I regard this argument as beautiful. For the final theorem, my role
was comparable to that of a senior scientist suggesting a suitable
series of problems and readings to a brilliant postdoctoral
researcher. Following my suggestions for the organization, ChatGPT
prepared the first draft, which I subsequently revised with its
assistance. I have checked the mathematical arguments and take full
responsibility for the paper. The numerical certificates imported
from \cite{OpenAIperiodic,OpenAIboundary} have not been independently
recomputed here.

\section{Definitions and main results}\label{S:results}

\subsection{The infinite chain and finite chains}

Let us consider a chain of $S=1$ spins on $\bbZ$. Associate $\bbC^3$
with each site, and write $\spn_j=(\Sx_j,\Sy_j,\Sz_j)$ for the spin
operators, normalized by $\spn_j^2=2\one$. Put $Q_j=(\Sz_j)^2$, so
that $0\le Q_j\le\one$. For a finite $\Lambda\subset\bbZ$, let
$\calA_\Lambda=\bigotimes_{j\in\Lambda}M_3(\bbC)$.
The local algebra $\Aloc$ is the union of these matrix algebras, with
the usual identifications by tensoring with the identity, and the
$C^*$-algebra $\calA$ is its norm completion. We identify a finite
interval $[s,t]$ of the chain with $\{s,s+1,\ldots,t\}$, where
$s,t\in\bbZ$.

We study the $J$-$D$ model with the formal infinite-chain Hamiltonian
\eq
H(D)=\sum_{j\in\bbZ}\{\spn_j\cdot\spn_{j+1}+DQ_j\},
\qquad D\in\bbR.
\lb{Hbulk}
\en
We have set the exchange coupling $J=1$, so $D$ is the only bulk
parameter. For a local operator $A$, the commutator $[H(D),A]$ is a
finite sum. A state $\omega$ on $\calA$ is a ground state if
\eq
\omega\bigl(A^*[H(D),A]\bigr)\ge0\qquad(A\in\Aloc).
\lb{GSdefinition}
\en
For the present finite-range interaction this is the standard
operator-algebraic ground-state condition; the local algebra is a core
for the generator of the dynamics \cite{BR2,TasakiBook}.
Let $\tau_x$ denote translation by $x\in\bbZ$, so
$\tau_x(S_j^\alpha)=S_{j+x}^\alpha$. Translation invariance means
$\omega\circ\tau_x=\omega$ for every integer $x$.
Our main task is to prove uniqueness in this class, without additional
assumptions on the state.

We next define finite-chain Hamiltonians with three boundary
conditions. For an integer $N\ge2$, define the centered interval
with $N$ sites by
\eq
\Lambda_N=
\begin{cases}
 \{-N/2,\ldots,N/2-1\},&N\text{ even},\\
 \{-(N-1)/2,\ldots,(N-1)/2\},&N\text{ odd}.
\end{cases}
\lb{centeredinterval}
\en
We shall also write
\eq
\Lambda_N=\{\ell_N,\ldots,r_N\}.
\lb{ellNrN}
\en
For every $N\ge4$, we define the periodic Hamiltonian by
\eq
H_{N}^{\per}(D)
 =\sum_{j\in\Lambda_N}
 \{\spn_j\cdot\spn_{j+1}+DQ_j\},
 \qquad \spn_{r_N+1}=\spn_{\ell_N},
 \lb{Hper}
 \en
and, for every $N\ge2$ and a fixed field strength $h_0>0$,
define the open-chain Hamiltonians by
 \eqg
 H_{N}^{\fld}(D)
 =\sum_{j=\ell_N}^{r_N-1}\spn_j\cdot\spn_{j+1}
   +D\sum_{j\in\Lambda_N}Q_j-h_0(\Sz_{\ell_N}+\Sz_{r_N}),
 \lb{Hfield}\\
 H_{N}^{\free}(D)
 =\sum_{j=\ell_N}^{r_N-1}\spn_j\cdot\spn_{j+1}
   +D\sum_{j\in\Lambda_N}Q_j.
   \lb{Hfree}
 \eng
In all three constructions, the endpoints or closing bond recede
from every fixed neighborhood of the origin as $N\uparrow\infty$.

We are primarily interested in the even periodic chain
$H_{2L}^{\per}(D)$, with $L\ge2$, and the odd open chain
$H_{2L+1}^{\fld}(D)$, with $L\ge1$.
In the latter, the same magnetic field in the $\mathrm z$ direction
acts at both endpoints; the single-ion term is present there as well.
This is the boundary condition used in \cite{Tasaki2025} to study
the topological character of the ground state.
The OpenAI papers \cite{OpenAIperiodic,OpenAIboundary} establish
uniform gaps above unique ground states of $H_{2L}^{\per}(0)$ and
$H_{2L+1}^{\fld}(0)$, respectively, at all sufficiently large $L$;
the latter result uses $h_0=3/5$. The same finite-chain criteria
remain valid for sufficiently small $|D|$ by continuity of the
finitely many traces entering them. We give the details in
Appendix~\ref{S:numerics}; continuity of individual eigenvalues
at each fixed length would not by itself imply a uniform gap.

We also study $H_N^{\free}(D)$ at both parities. For small $|D|$,
this Hamiltonian is expected to have four nearly degenerate
lowest-energy states. At $D=0$ and even $N$, these consist of the
singlet ground state and the Kennedy triplet; see \cite{TasakiBook}.
The near degeneracy is associated with an effective $S=1/2$ spin
near each endpoint. The fields in \rlb{Hfield} are introduced to
remove this freedom at the endpoints. All four low-lying states
are expected to have the same infinite-chain limit. Our result
below concerns arbitrary choices of actual ground vectors and
does not assume this four-state description.

Uniqueness is known independently for some of these finite chains.
In particular, $H_{2L}^{\per}(D)$ has a unique ground state for
every real $D$ and every $L\ge2$; see, e.g.,
\cite[Theorem~2.4]{TasakiBook}. We use this classical result to
obtain translation-invariant ground states of the infinite chain
without using a gap estimate. The estimates below also establish
the finite-chain uniqueness needed here, together with a uniform
gap, for sufficiently small $|D|$ and sufficiently large lengths.
Whenever the ground eigenspace is one-dimensional, choose
normalized ground vectors $\Phi_{2L,D}^{\per}$ and
$\Phi_{2L+1,D}^{\fld}$, and define
\eqa
\omega_{2L,D}^{\per}(A)
 &=\langle\Phi_{2L,D}^{\per},A\Phi_{2L,D}^{\per}\rangle,
 \lb{omegaper}\\
\omega_{2L+1,D}^{\fld}(A)
 &=\langle\Phi_{2L+1,D}^{\fld},A\Phi_{2L+1,D}^{\fld}\rangle.
 \lb{omegafield}
\ena
Each expression is defined as soon as the support of $A$ lies in the
corresponding interval. In particular, every $A\in\Aloc$ can be used
once $L$ is large enough.

The proof also uses finite-temperature Gibbs states.
For $b\in\{\per,\fld,\free\}$, inverse temperature $\beta>0$, and an
operator $A$ supported in $\Lambda_N$, define
\eq
\rho_{N,\beta,D}^{b}(A)
 =\frac{\Tr_{\mathcal H_N}\bigl(Ae^{-\beta H_N^b(D)}\bigr)}
        {\Tr_{\mathcal H_N}e^{-\beta H_N^b(D)}},
\qquad
\mathcal H_N=\bigotimes_{j\in\Lambda_N}\bbC^3.
\lb{Gibbsdefinition}
\en
Whenever $H_N^b(D)$ has a unique ground state, write its
expectation as $\omega_{N,D}^b$, consistently with
\rlb{omegaper} and \rlb{omegafield}. At that fixed length,
\eq
\omega_{N,D}^{b}(A)=\lim_{\beta\uparrow\infty}\rho_{N,\beta,D}^{b}(A).
\lb{fixedvolumecold}
\en
This elementary fixed-length limit does not by itself control the
joint increase of chain length and inverse temperature needed here.

\subsection{The main theorem and the three finite-chain limits}

We first state the uniqueness theorem for the infinite chain.
\begin{theorem}[Translation-invariant uniqueness]
\label{T:main}
There are constants $D_0,\gamma_0>0$ such that, for every
$D\in[-D_0,D_0]$, $H(D)$ has exactly one translation-invariant
ground state $\omega_D$. Uniqueness is asserted among all
translation-invariant states satisfying \rlb{GSdefinition}.
No spectral-gap or spin-rotation symmetry assumption is imposed
on a competing state.

The state $\omega_D$ is pure and is a locally-unique gapped ground
state: for every $A\in\Aloc$,
\eq
\omega_D\bigl(A^*[H(D),A]\bigr)
\ge\gamma_0\{\omega_D(A^*A)-|\omega_D(A)|^2\}.
\lb{localgap}
\en
It is invariant under spatial reflections, all uniform rotations about
the $\mathrm z$ axis, and the uniform $\pi$ rotation about the
$\mathrm x$ axis. At $D=0$ it is invariant under all spin rotations.
For every $A\in\calA$, the map $D\mapsto\omega_D(A)$ is continuous.
\end{theorem}

Here ``locally-unique gapped ground state'' has the meaning used in
\cite{TasakiLSMReview,TasakiEquivalence}. Setting $\omega_D(A)=0$ in
\rlb{localgap} gives the usual formulation. Conversely, that
formulation applied to $A-\omega_D(A)\one$ gives \rlb{localgap}.
This property is not being used as a synonym for uniqueness among
all ground states. The additional uniqueness proved in the theorem
is restricted to translation-invariant states.
The theorem is proved in Sections~\ref{S:transfer}--\ref{S:uniqueness}.

We next state the convergence theorem for periodic chains.
\begin{theorem}[Convergence of periodic-chain ground states]
\label{T:periodic}
With the constants and states of Theorem~\ref{T:main}, the ground
state of $H_{2L}^{\per}(D)$ is unique for every $L\ge2$ and
$|D|\le D_0$, and has an energy gap greater than $\gamma_0$ for
all sufficiently large $L$, with a threshold independent of $D$.
Moreover,
\eq
\lim_{L\uparrow\infty}\omega_{2L,D}^{\per}(A)
=\omega_D(A),\qquad A\in\Aloc,
\lb{PBClimit}
\en
uniformly for $|D|\le D_0$. There is $c>0$ such that, for every fixed
finite interval $I$, a constant $C_I>0$ satisfies
\eq
\sup_{|D|\le D_0}\ \sup_{\substack{A\in\calA_I\\\norm A\le1}}
|\omega_{2L,D}^{\per}(A)-\omega_D(A)|
\le C_I\sqrt L\,e^{-c\sqrt L}
\lb{localratePBC}
\en
for all sufficiently large $L$.
\end{theorem}

The supremum over $A$ in \rlb{localratePBC} is the trace norm of the
difference of the reduced density matrices on $I$.
Let us first explain why \rlb{PBClimit} follows once
translation-invariant uniqueness is established.
By the classical result recalled above, $H_{2L}^{\per}(D)$ has a
unique ground state for every $L\ge2$. Since cyclic translation
commutes with the Hamiltonian, it multiplies a normalized ground
vector by a phase. Its expectation functional is therefore invariant
under cyclic translation. On the support of a fixed local observable,
this agrees with the ordinary translation when $L$ is sufficiently
large. Compactness of the finite local state spaces, followed by a
diagonal subsequence, gives a limit of a subsequence of these
expectation functionals. Every such limit is translation invariant.
It is also a ground state: the finite-chain ground-state inequality
passes to the limit, since the closing bond eventually commutes with
each fixed local operator. This proves existence of at least one
translation-invariant ground state independently of
Theorem~\ref{T:main}. Once uniqueness in that theorem is established,
every subsequential limit must equal $\omega_D$, so the entire
periodic sequence converges to $\omega_D$.
Uniformity in $D$ is proved in Section~\ref{S:bulkconsequences}.
The quantitative bound \rlb{localratePBC} requires further estimates;
it is proved in Section~\ref{S:quantitativerate} using only the
periodic estimates.

The following convergence theorem is important for the application
to the topological properties of the ground state in
\cite{Tasaki2025,TasakiEquivalence}.
\begin{theorem}[Convergence with equal endpoint fields]
\label{T:boundary}
There is a field strength $h_0>0$ such that, with the constants
and states of Theorem~\ref{T:main}, the ground state of
$H_{2L+1}^{\fld}(D)$ is unique and has an energy gap greater than
$\gamma_0$ for every sufficiently large $L$, uniformly for
$|D|\le D_0$. Moreover,
\eq
\lim_{L\uparrow\infty}\omega_{2L+1,D}^{\fld}(A)
=\omega_D(A),\qquad A\in\Aloc,
\lb{fieldlimit}
\en
uniformly for $|D|\le D_0$. There is $c>0$ such that, for every fixed
finite interval $I$, a constant $C_I>0$ satisfies
\eq
\sup_{|D|\le D_0}\ \sup_{\substack{A\in\calA_I\\\norm A\le1}}
|\omega_{2L+1,D}^{\fld}(A)-\omega_D(A)|
\le C_I\sqrt L\,e^{-c\sqrt L}
\lb{localrateField}
\en
for all sufficiently large $L$.
\end{theorem}

The limit \rlb{fieldlimit} likewise follows from
Theorem~\ref{T:main} if every subsequential limit of these
ground-state expectations is translation invariant. Unlike the
periodic case, translation invariance is not an exact finite-chain
symmetry. Proving its recovery in the limit requires the estimates
of Sections~\ref{S:fieldtransfer}--\ref{S:limit}. The constants
$c$ in Theorems~\ref{T:periodic} and \ref{T:boundary} can be
chosen to be the same by decreasing either one if necessary.

Finally, there is a convergence theorem for open chains without
endpoint magnetic fields, a setting not covered by the two cited
OpenAI gap theorems. Let $E_N^{\free}(D)$ be the ground-state
energy of $H_N^{\free}(D)$.
\begin{theorem}[Convergence without endpoint fields]\label{T:free}
For every fixed $D$ such that $|D|\le D_0$ and every $A\in\Aloc$,
\eq
\lim_{N\uparrow\infty}
\sup_{\substack{\psi\in\ker(H_N^{\free}(D)-E_N^{\free}(D))\\
                \norm\psi=1}}
\abs{\langle\psi,A\psi\rangle-\omega_D(A)}=0.
\lb{freelimit}
\en
Both even and odd lengths are included.
The supremum is over all normalized ground vectors of
$H_N^{\free}(D)$. No uniqueness or gap uniform in $N$ is assumed
for these finite-chain ground states.
\end{theorem}

Theorems~\ref{T:main}, \ref{T:periodic}, and \ref{T:free} use
only the periodic numerical inputs of \cite{OpenAIperiodic}.
Theorem~\ref{T:boundary} also uses the boundary-field numerical
inputs of \cite{OpenAIboundary}. One simultaneous choice, verified
in Appendix~\ref{S:numerics}, is $h_0=3/5$,
$D_0=5\times10^{-9}$, and $\gamma_0=(\log10)/392$.
We keep the constants abstract in the proofs. Although only the
single-ion anisotropy is developed in detail, the argument also
extends to sufficiently small exchange anisotropy.\footnote{One
may replace $\spn_j\cdot\spn_{j+1}$ by
$S_j^{\mathrm x}S_{j+1}^{\mathrm x}
+S_j^{\mathrm y}S_{j+1}^{\mathrm y}
+\Delta S_j^{\mathrm z}S_{j+1}^{\mathrm z}$ with
$|\Delta-1|$ sufficiently small, retaining the symmetries
preserved by that interaction. We do not develop this extension here.}

\subsection{Organization and the additional arguments}

Section~\ref{S:transfer} constructs the periodic spatial transfer and
local insertions. Section~\ref{S:criterion} proves the periodic
concentration estimates and finite-chain gaps from two finite
purities. Section~\ref{S:uniqueness} proves
Theorem~\ref{T:main}: the small free-energy correction and the separated
exponential-moment bound identify every translation-invariant ground
state without assuming a zero-temperature limit in advance.
The qualitative periodic limit in Theorem~\ref{T:periodic} follows
there as well; its quantitative bound is completed in
Section~\ref{S:quantitativerate}.

Sections~\ref{S:fieldtransfer}--\ref{S:longchains} establish the
finite-chain endpoint-field comparison in
Proposition~\ref{P:finiteapproximation}.
Section~\ref{S:limit} proves that every limit of a subsequence of
the ground-state expectations with endpoint fields is translation
invariant, and then applies Theorem~\ref{T:main} to prove
Theorem~\ref{T:boundary}. Section~\ref{S:free} proves
Theorem~\ref{T:free}. The appendix verifies the numerical premises.

\medskip\noindent
{\em Additions relative to the two OpenAI papers.}\/ 
The periodic purity iteration and the coupled control of the boundary
overlap and opposite-field Gibbs purity are inherited from
\cite{OpenAIperiodic,OpenAIboundary}. The principal additional argument
is the use of repeated local insertions, the low-temperature free-energy
estimate, and a scalar Jensen--Golden--Thompson comparison to prove
translation-invariant uniqueness. We also incorporate the single-ion
evolution into the kernel, give an initial overlap test using ordinary
finite-chain traces, and obtain quantitative local expectations.
The elementary sign argument for the $\pi$ rotation about the $\mathrm x$ axis does
not require the cyclic permutation of spin axes used in the isotropic
boundary proof. Retaining the length dependence of the physical
purity permits lengths of order $\beta_j^2$; in the equal-field case
this extends the estimate in the proof of
\cite[Proposition~6.2]{OpenAIboundary}, not the boundary recurrence
itself. Finally, for free ends we use monotonicity of Gibbs purity
and extremality of the pure limiting state rather than purification
of the finite-chain Gibbs density matrix. These statements describe
additions relative to the two cited OpenAI papers, not priority claims
for the general transfer, variational, or convexity principles.

\section{Gibbs states and their spatial transfer representation for periodic chains}
\label{S:transfer}

The proof of translation-invariant uniqueness uses the Gibbs
expectations defined in \rlb{Gibbsdefinition}.
We first derive exact expressions for periodic chains.

\subsection{Treatment of the \texorpdfstring{$D$}{D}-term in the bond expansion}

We reserve the notation
\eq
Z_{N,D}(\beta)=\Tr_{\mathcal H_N}e^{-\beta H_N^{\per}(D)}
\lb{Zdefinition}
\en
for the \emph{periodic} partition function. We recall the bond-history space used in \cite{OpenAIperiodic,TasakiCT}.
An element $x\in\Omega_\beta$ is either the empty history
$\varnothing$ or a finite ordered list
$x=((t_1,\alpha_1),\ldots,(t_r,\alpha_r))$, where
$0<t_1<\cdots<t_r<\beta$ and each mark $\alpha_a$ belongs to
$\{\mathrm x,\mathrm y,\mathrm z\}$.
The measure $\mu_\beta$ is defined by
\eq
\begin{split}
\int_{\Omega_\beta}d\mu_\beta(x)\,F(x)
={}&F(\varnothing)\nobreak\\
 &+\sum_{r=1}^{\infty}
  \sum_{\alpha_1,\ldots,\alpha_r\in\{\mathrm x,\mathrm y,\mathrm z\}}
  \int_{0<t_1<\cdots<t_r<\beta}dt_1\cdots dt_r\,
  F\bigl((t_1,\alpha_1),\ldots,(t_r,\alpha_r)\bigr)
\end{split}
\lb{historymeasure}
\en
for nonnegative measurable $F$, and hence for every integrable $F$.
Thus the empty history has mass one, and
$\mu_\beta(\Omega_\beta)=\sum_{r=0}^{\infty}3^r\beta^r/r!=e^{3\beta}$.
In particular, this is not a probability measure.
Put $\calK_\beta=L^2(\Omega_\beta,\mu_\beta;\bbC)$.
Each history describes events on a single bond. A full space-time
configuration assigns one such history to each bond; the bond index
is not part of an individual history.
Put $G^\alpha=iS^\alpha$ and $Q=(\Sz)^2$, so
$\norm{G^\alpha}=1$ and
$-\spn_j\cdot\spn_{j+1}=\sum_\alpha G_j^\alpha G_{j+1}^\alpha$.

The only change from the expansion discussed in the CT note is that
we keep the single-ion evolution exact between successive bond events. Here is a derivation
from the Lie-product formula. On a fixed periodic chain, set
\eq
\mathcal Q_N=\sum_{j\in\Lambda_N}Q_j,
\qquad
V_N=\sum_{j\in\Lambda_N}\sum_{\alpha=\mathrm x,\mathrm y,\mathrm z}
 G_j^\alpha G_{j+1}^\alpha,
\qquad H_N^{\per}(D)=D\mathcal Q_N-V_N,
\lb{QVdecomposition}
\en
with periodic identification in $V_N$. For $\delta=\beta/n$,
\eq
e^{-\beta H_N^{\per}(D)}
 =\lim_{n\uparrow\infty}
 \bigl(e^{-\delta D\mathcal Q_N}e^{\delta V_N}\bigr)^n
 =\lim_{n\uparrow\infty}
 \bigl(e^{-\delta D\mathcal Q_N}(\one+\delta V_N)\bigr)^n.
\lb{Lieproduct}
\en
Both limits are in operator norm. The second follows by a telescoping
estimate: replacing $e^{\delta V_N}$ by $\one+\delta V_N$ has error
$O(\delta^2)$ per factor, and the number of factors is $n$.
Here $N,\beta,D$ are fixed while this limit is taken.

Expand the last product by choosing the slices in which $V_N$ occurs.
For a fixed number $r$ of such slices, their ordered positions give a
Riemann sum over $0<t_1<\cdots<t_r<\beta$. Consequently,
\eq
\begin{split}
e^{-\beta H_N^{\per}(D)}
 =\sum_{r=0}^{\infty}\int_{0<t_1<\cdots<t_r<\beta}
 &dt_1\cdots dt_r\,
 e^{-t_1D\mathcal Q_N}V_Ne^{-(t_2-t_1)D\mathcal Q_N}
 \cdots\\[-1mm]
 &\hspace{18mm}\cdots V_Ne^{-(\beta-t_r)D\mathcal Q_N}.
\end{split}
\lb{orderedDexpansion}
\en
The $r=0$ term is $e^{-\beta D\mathcal Q_N}$.
To justify summing the limits, put $D_-=\max\{0,-D\}$.
At order $r$ the sum of the operator-norm bounds, after also expanding
$V_N$ into its $3N$ bond terms, is at most
$e^{\beta ND_-}(3N\beta)^r/r!$.
These bounds are summable. Thus this is an ordered expansion in the
\emph{bond terms}, not a Taylor expansion in the full Hamiltonian.
It uses no commutation of $Q_j$ with the bond operators.

Regroup the events by assigning to each bond its ordered list of events.
If two histories $x,y$ incident on a site merge into the chronological
list $(t_1,\alpha_1),\ldots,(t_r,\alpha_r)$, define its single-site word
and kernel by
\eqa
W_{\beta,D}(x,y)
 &=e^{-Dt_1Q}G^{\alpha_1}e^{-D(t_2-t_1)Q}
       \cdots G^{\alpha_r}e^{-D(\beta-t_r)Q},
 \lb{Dword}\\
k_{\beta,D}(x,y)&=\tr W_{\beta,D}(x,y).
 \lb{Dkernel}
\ena
For the empty merged history, $W_{\beta,D}(\varnothing,\varnothing)
=e^{-\beta DQ}$ and $k_{\beta,D}(\varnothing,\varnothing)=1+2e^{-\beta D}$.
The increasing-time order is from left to right, as in the CT note.
Coincident event times form a null set and do not affect any integral.
At $D=0$, \rlb{Dkernel} is precisely the CT kernel.

Let $\Theta$ be the usual spin-one time-reversal antiunitary, with
$\Theta^{-1}S^\alpha\Theta=-S^\alpha$. Antilinearity then gives
$\Theta^{-1}G^\alpha\Theta=G^\alpha$, and
$\Theta^{-1}Q\Theta=Q$. A word fixed by $\Theta$ has real trace,
since antiunitary conjugation complex-conjugates its trace. This applies also with the factors $e^{-tDQ}$ inserted. Exchanging $x$ and $y$ leaves the merged list unchanged, so
\eq
k_{\beta,D}(x,y)=k_{\beta,D}(y,x)\in\bbR,
\qquad
|k_{\beta,D}(x,y)|\le3e^{\beta D_-}.
\lb{kernelproperties}
\en
The integral operator
\eq
(K_{\beta,D}f)(x)
 =\int_{\Omega_\beta}k_{\beta,D}(x,y)f(y)\,d\mu_\beta(y)
\lb{Kdefinition}
\en
is therefore real, self-adjoint, and Hilbert--Schmidt, with
$\norm{K_{\beta,D}}_{\rm HS}^2\le9e^{2(3+D_-)\beta}$.
Regrouping the physical trace gives the cyclic integral of $N$
site kernels. This is the operator trace: the kernel of $K^{N-1}$ is
its iterated integral, and for Hilbert--Schmidt kernels $a,b$,
$\Tr(AB)=\int a(x,y)b(y,x)\,d\mu_\beta(x)d\mu_\beta(y)$.
The bounded kernel and finite measure justify every integration.
Consequently,
\eq
Z_{N,D}(\beta)=\Tr_{\calK_\beta}K_{\beta,D}^{\,N},
\qquad N\ge4.
\lb{Ztransfer}
\en
This holds at both parities. The operator $K_{\beta,D}$ need not be
positive. We have therefore recovered the CT transfer representation
with the single-ion term included. No statement about a leading
eigenvalue has yet been used.

\subsection{The space-time representation for expectation values}
\label{S:blockconstruction}

Let $A$ be an operator supported on a block $I=[s,t]$ of $m=t-s+1$ sites,
strictly smaller than the ring. Start from
$\Tr_{\mathcal H_N}(Ae^{-\beta H_N^{\per}(D)})$ and apply the same
expansion \rlb{orderedDexpansion}. The operator $A$ sits at imaginary time $0$, immediately to the left
of all the ordered single-site words. Figure~\ref{fig:periodic} shows the ordinary
trace and the trace with this insertion.

\begin{figure}[tbp]
\centering
\includegraphics[width=0.85\textwidth]{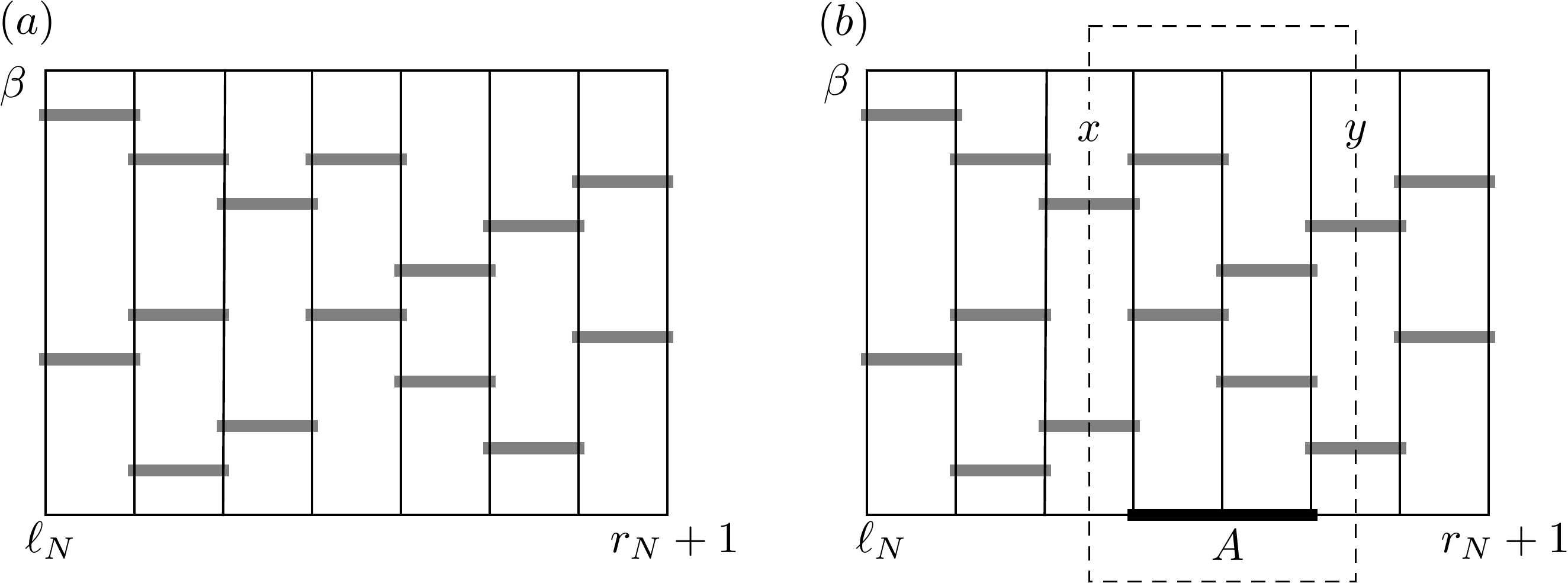}
\caption{The periodic bond-history representation. Panel (a) represents
$Z_{N,D}(\beta)$; panel (b) represents
$\Tr_{\mathcal H_N}(Ae^{-\beta H_N^{\per}(D)})$.
Vertical lines are spin sites and gray horizontal segments are bond
events; the spin-component labels are suppressed. The single-ion
factors act along the intervals between events and are not drawn
separately. The displayed spatial endpoints are periodically
identified: the line labeled $r_N+1$ is the periodic copy of site
$\ell_N$, not an extra spin. The drawing is schematic and the transfer
identity applies at either parity.
The black segment is $A$ at time $0$. The histories $x,y$ crossing the
dashed sides are retained as the two arguments of the insertion kernel;
the histories inside the block are integrated out. The dashed lines
are not physical boundaries. The physical trace identifies the top
and bottom in imaginary time, with the insertion of $A$ in panel (b).}
\label{fig:periodic}
\end{figure}

Outside $I$, each site again contributes the trace of its word
$W_{\beta,D}$. Inside $I$, an arbitrary $A$ may couple the sites, so
we must keep a single trace over the whole block. Retain the histories
$x,y$ on its two external bonds, let $x_1,\ldots,x_{m-1}$ be its
internal bond histories, and set $x_0=x$, $x_m=y$. Define
\eq
\begin{split}
b_{A,\beta,D}(x,y)
 =\int\prod_{a=1}^{m-1}d\mu_\beta(x_a)\,
 \Tr_{(\bbC^3)^{\otimes m}}
 \left[A\bigotimes_{a=1}^{m}W_{\beta,D}(x_{a-1},x_a)\right],
\end{split}
\lb{blockinsertion}
\en
where every history measure is the measure defined above.
For $m=1$ there is no internal integration. Let
$\mathcal B_{A,\beta,D}$ be the integral operator on $\calK_\beta$
with kernel $b_{A,\beta,D}$. This construction is linear in $A$ and
does not require $A$ to be a product of one-site operators.
It replaces \emph{all $m$ site kernels in the block}, not just one
factor. In particular,
\eq
\mathcal B_{\one,\beta,D}=K_{\beta,D}^{\,m}.
\lb{Bidentity}
\en
For a general $A$, the operator $\mathcal B_{A,\beta,D}$ need not be
self-adjoint or positive.

There is no convergence difficulty in this regrouping. Each word has
norm at most $e^{\beta D_-}$. Bounding the block trace by
$3^m\norm A$ times the product of these norms, and integrating the
$m-1$ internal histories, bounds the kernel by
$3^m\norm A e^{m\beta D_-+3(m-1)\beta}$.
Integrating the two external histories in the Hilbert--Schmidt norm
therefore gives
\eq
\norm{\mathcal B_{A,\beta,D}}
 \le\norm{\mathcal B_{A,\beta,D}}_{\rm HS}
 \le3^m\norm A\,e^{(3+D_-)m\beta}.
\lb{Bnorm}
\en
The absolute sum for the full traced expansion is at most
$3^N\norm A e^{\beta N(3+D_-)}$.

Closing the remaining $N-m$ site kernels into a periodic chain gives
\eq
\Tr_{\mathcal H_N}\bigl(Ae^{-\beta H_N^{\per}(D)}\bigr)
 =\Tr_{\calK_\beta}
  \bigl(\mathcal B_{A,\beta,D}K_{\beta,D}^{N-m}\bigr),
\qquad N\ge\max\{4,m+1\}.
\lb{perunnormalized}
\en
The product on the right is trace class: $\mathcal B_A$ is
Hilbert--Schmidt and so is $K^{N-m}$ for $N-m\ge1$.
Dividing by \rlb{Ztransfer}, we obtain the exact finite-chain identity
\eq
\rho_{N,\beta,D}^{\per}(A)
 =\frac{\Tr\bigl(\mathcal B_{A,\beta,D}K_{\beta,D}^{N-m}\bigr)}
        {\Tr K_{\beta,D}^{N}}.
\lb{perGibbsraw}
\en
For $A=\one$, \rlb{Bidentity} makes the numerator equal to the
denominator. Both traces on the right act on the history Hilbert
space; the equality with the physical Gibbs expectation is the result
of the bond expansion. The local observable has now been incorporated
into a block operator, so the remaining calculation concerns only the
transfer factors outside its support.

\subsection{The sign of a dominant transfer eigenvalue}
\label{S:leadingline}

We next explain what a single dominant eigenvector of the spatial transfer operator would imply.
Its existence will be proved at the required temperatures by purity
concentration; it is not being assumed for every $\beta$.
When $K_{\beta,D}$ has a simple eigenvalue whose modulus is strictly
larger than that of every other eigenvalue,
write $\lambda_{\beta,D}=\norm{K_{\beta,D}}>0$ for that modulus and
$P_{\beta,D}$ for the projection onto its one-dimensional eigenspace.

\begin{lemma}[Positivity of the leading eigenvalue]\label{L:positiveleading}
Fix $\beta>0$ and $D\in\bbR$, and suppose that $K_{\beta,D}$ has a
simple eigenvalue of strictly largest modulus. With the notation just
introduced, the leading eigenvalue is $+\lambda_{\beta,D}$.
\end{lemma}
\begin{proof}
Abbreviate $K=K_{\beta,D}$, $\lambda=\lambda_{\beta,D}$, and
$P=P_{\beta,D}$. Let the leading eigenvalue be $\epsilon\lambda$,
with $\epsilon\in\{-1,1\}$, and let $\varrho<1$ bound the moduli
of all the other eigenvalues divided by $\lambda$.
The spectral expansion of \rlb{Ztransfer} gives
\eq
\frac{Z_{N,D}(\beta)}{\lambda^N}
 =\epsilon^N+\sum_{\xi\ne\epsilon\lambda}(\xi/\lambda)^N.
\lb{signlambda}
\en
The sum over the other eigenvalues tends to zero. To see this, use
\eq
\left|\sum_{\xi\ne\epsilon\lambda}(\xi/\lambda)^N\right|
 \le\varrho^{N-4}
        \sum_{\xi\ne\epsilon\lambda}|\xi/\lambda|^4
 \longrightarrow0,
\lb{signremainder}
\en
where the sum of fourth powers is finite and includes multiplicities. If $\epsilon=-1$, the expression would
therefore tend to $-1$ through odd $N$, contradicting positivity of
the physical partition function. Thus $\epsilon=+1$. The spectral decomposition also gives
$\norm{\lambda^{-N}K^N-P}\le\varrho^N$, which proves convergence
to $P$ in operator norm as $N\uparrow\infty$.

\end{proof}

This lemma uses positivity of the physical partition functions at odd
as well as even lengths, although the ground-state sequence studied
in the periodic case has even length.

\subsection{The periodic Gibbs limit at a fixed temperature}
\label{S:strategy}

For an operator $A$ supported on $m$ sites, define, whenever the
simple dominant eigenvalue exists,
\eq
\nu_{\beta,D}(A)=\frac{\Tr(P_{\beta,D}\mathcal B_{A,\beta,D})}
                            {\lambda_{\beta,D}^{m}}.
\lb{nufixedbeta}
\en
\begin{lemma}[Periodic Gibbs limit]\label{L:fixedtemperature}
Fix $\beta>0$ and $D\in\bbR$. If $K_{\beta,D}$ has a simple
eigenvalue of strictly largest modulus, then
\eq
\lim_{N\uparrow\infty}\rho_{N,\beta,D}^{\per}(A)=\nu_{\beta,D}(A)
\qquad(A\in\Aloc).
\lb{fixedtemperaturelimit}
\en
The expression defines a translation-invariant state on $\calA$ and
is independent of the containing interval used for its insertion.
\end{lemma}
\begin{proof}
Lemma~\ref{L:positiveleading} gives a positive leading eigenvalue.
The argument in \rlb{signremainder}, with absolute values inside the
sum, proves
\eq
\left\|\lambda_{\beta,D}^{-n}K_{\beta,D}^{n}-P_{\beta,D}\right\|_1
 \longrightarrow0\qquad(n\uparrow\infty).
\lb{fixedtemperaturepower}
\en
Divide the numerator and denominator of \rlb{perGibbsraw} by
$\lambda_{\beta,D}^N$. The numerator converges to
\rlb{nufixedbeta} and the denominator to one. Positivity,
normalization, consistency on nested local algebras, and translation
invariance pass from the finite periodic Gibbs expectations to their
limit. In particular $|\nu_{\beta,D}(A)|\le\norm A$, so the
functional extends by norm continuity to $\calA$.
\end{proof}

The factor $\lambda_{\beta,D}^{-m}$ in \rlb{nufixedbeta} records the $m$ transfer factors
replaced by the insertion block. No projection was substituted into
an exact finite-chain expression. We have defined a state at fixed
temperature, but have not assumed a limit as $\beta\uparrow\infty$.
The next section supplies the spectral estimates at a sequence of
temperatures; Section~\ref{S:uniqueness} uses them to prove uniqueness
and, as a consequence, the zero-temperature limit.

\section{The periodic criterion and concentration estimates}
\label{S:criterion}

\subsection{Two finite purities and the doubling scales}

For even $k\ge4$ and $N\ge4$, define
\eq
S_D(k,\beta)=\frac{Z_{2k,D}(\beta)}{Z_{k,D}(\beta)^2},
\qquad
T_D(N,\beta)=\frac{Z_{N,D}(2\beta)}{Z_{N,D}(\beta)^2}.
\lb{STdef}
\en
The first is the purity of the probability list
$|\xi_a|^k/Z_{k,D}(\beta)$ formed from the transfer eigenvalues,
counted with multiplicity. The second is the physical Gibbs purity,
with all physical eigenvectors counted, and is defined at either
parity. Both lie in $(0,1]$.

Fix an even $k_0\ge4$, $\beta_0>0$, $0<\eta\le1/100$, and
$D_0>0$. The hypothesis for the main theorem is only
\eq
S_D(k_0,\beta_0)\ge1-\eta,
\qquad T_D(2k_0,\beta_0)\ge1-\eta
\qquad(|D|\le D_0).
\lb{periodiccriterion}
\en
It involves three ordinary finite periodic traces. Strict margins at
$D=0$ persist for small $|D|$ by continuity; Appendix~\ref{S:numerics}
gives explicit bounds. No open-chain quantity occurs in this criterion.
Set
\eq
\gamma_0=\frac{\log(1/(10\eta))}{\beta_0},
\qquad q=(10\eta)^{1/k_0}<1,
\lb{gamma0}
\en
and
\eq
\beta_j=2^j\beta_0,\qquad k_j=2^jk_0,\qquad j=0,1,\ldots.
\lb{scales}
\en
We prove the main theorem conditionally on \rlb{periodiccriterion};
its verification will involve only periodic numerical inputs.

\subsection{The coupled purity iteration}

For a finite or countable probability list $(p_a)$ with purity
$P=\sum_a p_a^2$, the normalized squares have purity $P'$ satisfying
\eq
1-P'=\frac{2\sum_{a<b}p_a^2p_b^2}{P^2}
 \le\frac{2(\sum_{a<b}p_ap_b)^2}{P^2}
 =\frac{(1-P)^2}{2P^2}.
\lb{squaringpurity}
\en
All sums contain nonnegative terms. Thus $P'\ge h(P)$, where
\eq
h(p)=1-\frac{(1-p)^2}{2p^2},\qquad 0<p\le1.
\lb{hdefinition}
\en
The function $h$ is increasing. Squaring the spatial probabilities
doubles length; squaring the physical probabilities doubles inverse
temperature. Direct cancellation in \rlb{STdef} gives
\eqg
S_D(k,2\beta)
 =\frac{S_D(k,\beta)^2T_D(2k,\beta)}{T_D(k,\beta)^2},
\lb{coupledone}\\
T_D(4k,\beta)
 =\frac{T_D(2k,\beta)^2S_D(2k,2\beta)}{S_D(2k,\beta)^2}.
\lb{coupledtwo}
\eng
Both denominators are at most one. Apply \rlb{squaringpurity}
first in space and then in temperature to obtain the coupled
inequalities used below. This is the periodic mechanism of
\cite{OpenAIperiodic}, written in the notation of the CT note.

Fix $D\in[-D_0,D_0]$. All estimates below are uniform on this interval.
Recall the scales in \rlb{scales}; in particular,
$\beta_{j+1}=2\beta_j$ and $k_{j+1}=2k_j$.
We shall use two error sequences:
\eq
u_j=\frac16(6\eta)^{2^j},\qquad
x_j=\frac1{10}(10\eta)^{2^j}.
\lb{errorscales}
\en
Thus $u_0=x_0=\eta$, $u_{j+1}=6u_j^2$,
$x_{j+1}=10x_j^2$, and
\eq
u_j\le x_j.
\lb{errormatching}
\en
The first sequence controls the two periodic purities. The slightly
larger second sequence will also control the boundary estimates.

\begin{lemma}[The periodic purity bounds]\label{L:periodicpurity}
Under \rlb{periodiccriterion}, for every $j\ge0$,
\eq
S_D(k_j,\beta_j)\ge1-u_j,\qquad
T_D(2k_j,\beta_j)\ge1-u_j.
\lb{uperiodic}
\en
\end{lemma}
\begin{proof}
This is the iteration in the proof of Theorem~6.2 of the CT note,
with $u_0=\eta$ and $C=6$, applied to the present partition functions.
We recall the proof to fix the constants.
Equations \rlb{squaringpurity}--\rlb{coupledtwo} give
\eqa
S_D(2k,2\beta)
 &\ge h\bigl(S_D(k,\beta)^2T_D(2k,\beta)\bigr),\nl
T_D(4k,2\beta)
 &\ge h\bigl(T_D(2k,\beta)^2S_D(2k,2\beta)\bigr).
\lb{periodiciteration}
\ena
For $0\le u\le1/100$,
\eq
1-h((1-u)^3)
 \le\frac{9u^2}{2(1-3u)^2}\le6u^2\le u.
\lb{periodicstep}
\en
To obtain the first inequality, use $1-(1-u)^3\le3u$ and
$(1-u)^3\ge1-3u$ in the definition of $h$.
Suppose both purities at level $j$ are at least $1-u_j$.
The first update in \rlb{periodiciteration} and \rlb{periodicstep}
give $S_D(k_{j+1},\beta_{j+1})\ge1-6u_j^2=1-u_{j+1}$.
Since $u_{j+1}\le u_j$, this new spatial purity is also at least
$1-u_j$. In the second update we may therefore bound all three
factors inside $h$ from below by $1-u_j$, obtaining
$T_D(2k_{j+1},\beta_{j+1})\ge1-u_{j+1}$.
The hypotheses supply the case $j=0$, so induction proves the lemma.
\end{proof}

Both periodic purities are now controlled at every doubling scale,
with errors that decrease quadratically at each step.
The next step translates the first bound into explicit information
about the transfer eigenvalues; the second will relate consecutive
temperatures.

\subsection{The leading eigenvalue and the other eigenvalues}

We now translate the purity bound into a direct estimate on the
eigenvalues of $K_{\beta_j,D}$.

\begin{lemma}[Concentration of the transfer eigenvalues]\label{L:spectralconcentration}
Under \rlb{periodiccriterion}, for every $j\ge0$, $K_{\beta_j,D}$ has a simple positive eigenvalue
$\lambda_j$ whose modulus is strictly larger than that of every other
eigenvalue. Let $P_j$ be its orthogonal spectral projection, and let
$\varrho_j$ be the supremum of $|\xi|/\lambda_j$ over the other
eigenvalues $\xi$ (zero if there are no other nonzero eigenvalues).
Then
\eq
\varrho_j^{k_j}
 \le\frac{Z_{k_j,D}(\beta_j)-\lambda_j^{k_j}}{\lambda_j^{k_j}}
 \le d(u_j),\qquad
 d(u)=\frac{u}{2-3u}.
\lb{residualmoment}
\en
In particular, $\varrho_j<1$.
\end{lemma}
\begin{proof}
We first prove an elementary fact about a probability distribution
$(p_a)$ with $\sum_a p_a^2\ge1-u$, where $0\le u\le1/100$.
Its largest term $p_0$ satisfies
$\sum_a p_a^2\le p_0\sum_a p_a=p_0$, so $p_0\ge1-u>1/2$
and it is unique. Put $z=1-p_0\le u$.
Since $\sum_{a\ne0}p_a^2\le z^2$, we have
$1-u\le(1-z)^2+z^2$, or $2z(1-z)\le u$. Since $z\le u$, we have
$2z(1-u)\le2z(1-z)\le u$, and thus $z\le u/[2(1-u)]$.
Using the fact that $z/(1-z)$ is increasing on $[0,1)$, we obtain
\eq
\frac{1-p_0}{p_0}=\frac z{1-z}
 \le\frac{u/[2(1-u)]}{1-u/[2(1-u)]}
 =\frac{u}{2-3u}=d(u).
\lb{ddef}
\en

Apply this fact to
\eq
p_a=\frac{|\xi_a|^{k_j}}{Z_{k_j,D}(\beta_j)},
\qquad
\sum_a p_a=1,\qquad
\sum_a p_a^2=S_D(k_j,\beta_j)\ge1-u_j,
\lb{eigenvalueweights}
\en
where $(\xi_a)$ lists the nonzero eigenvalues of $K_{\beta_j,D}$,
with multiplicities. The unique largest term gives a simple eigenvalue
of largest modulus. Lemma~\ref{L:positiveleading} shows that this eigenvalue
is positive; call it $\lambda_j$.
The ratio in \rlb{ddef} is now
\eq
\frac{1-p_0}{p_0}
 =\frac{\sum_{\xi_a\ne\lambda_j}|\xi_a|^{k_j}}{\lambda_j^{k_j}}
 =\frac{Z_{k_j,D}(\beta_j)-\lambda_j^{k_j}}{\lambda_j^{k_j}}.
\lb{spectralremainder}
\en
The first inequality in \rlb{residualmoment} holds because any one
nonnegative summand is at most their sum. This proves both bounds.
\end{proof}

The sum bound in \rlb{residualmoment} controls traces, while the bound
on each individual eigenvalue controls further powers of the transfer
operator. Both will be used below. The proof also applies with any
even length $k$ and any $u\le1/100$ for which
$S_D(k,\beta_j)\ge1-u$.
The dependence of $\lambda_j$, $P_j$, and $\varrho_j$ on the fixed
parameter $D$ is understood; the original operator remains
$K_{\beta_j,D}$ throughout.

\subsection{Comparing the leading eigenvalues at two temperatures}

Both the uniqueness proof and the boundary arguments will compare
$\beta_j$ and $\beta_{j+1}=2\beta_j$.
For this purpose define the scalar ratio
\eq
r_j=\frac{\lambda_{j+1}}{\lambda_j^2}.
\lb{rdefinition}
\en

\begin{lemma}\label{L:temperatureeigenvalue}
Under \rlb{periodiccriterion}, for every $j\ge0$,
\eq
0<r_j\le1,\qquad
r_j^{k_{j+1}}\ge(1-u_j)(1-u_{j+1}).
\lb{rbound}
\en
\end{lemma}
\begin{proof}
To prove the upper bound, the physical Boltzmann weights give
$Z_{N,D}(2\beta_j)\le Z_{N,D}(\beta_j)^2$.
At either temperature, the simple leading eigenvalue and the summable
fourth spectral moment imply
$Z_{N,D}(\beta_j)^{1/N}\to\lambda_j$ through even $N$.
Taking $N$th roots therefore gives $\lambda_{j+1}\le\lambda_j^2$.

To prove the lower bound, list the eigenvalues of
$K_{\beta_{j+1},D}$ as $(\xi_a)$, with multiplicities.
The evenness of $k_{j+1}$ and $|\xi_a|\le\lambda_{j+1}$ give
\eq
\begin{split}
Z_{2k_{j+1},D}(\beta_{j+1})
 &=\sum_a|\xi_a|^{2k_{j+1}}\\
 &\le\lambda_{j+1}^{k_{j+1}}
       \sum_a|\xi_a|^{k_{j+1}}
 =\lambda_{j+1}^{k_{j+1}}Z_{k_{j+1},D}(\beta_{j+1}).
\end{split}
\lb{maximumeigenvalue}
\en
Dividing by the positive partition function and using the spatial
purity bound in \rlb{uperiodic} at level $j+1$, we obtain
\eq
\lambda_{j+1}^{k_{j+1}}
 \ge S_D(k_{j+1},\beta_{j+1})Z_{k_{j+1},D}(\beta_{j+1})
 \ge(1-u_{j+1})Z_{k_{j+1},D}(\beta_{j+1}).
\lb{leadingweightbound}
\en
Next use $\beta_{j+1}=2\beta_j$, $k_{j+1}=2k_j$, and the
physical-purity bound in \rlb{uperiodic}:
\eqa
\lambda_{j+1}^{k_{j+1}}
 &\ge(1-u_{j+1})T_D(2k_j,\beta_j)
                  Z_{k_{j+1},D}(\beta_j)^2\nl
 &\ge(1-u_{j+1})(1-u_j)Z_{k_{j+1},D}(\beta_j)^2\nl
 &\ge(1-u_{j+1})(1-u_j)\lambda_j^{2k_{j+1}}.
\lb{rproof}
\ena
The last step retains the leading eigenvalue in a sum of nonnegative
even powers. Division by $\lambda_j^{2k_{j+1}}$ proves the claim.
\end{proof}

The second inequality in \rlb{rbound} forces $r_j$ to be extremely
close to one: its $k_{j+1}$th power is bounded below by a number already
close to one. In Section~\ref{S:uniqueness} this will make the free-energy
correction very small. Later it will allow physical-purity estimates
to remain useful at lengths of order $\beta_j^2$.

\subsection{Uniform spatial contraction and local insertions}

Since $d(u_j)\le u_j\le x_j=q^{k_j}/10$, where the last equality
follows directly from \rlb{gamma0}, \rlb{scales}, and
\rlb{errorscales}, the individual-eigenvalue bound in
\rlb{residualmoment} implies $\varrho_j\le q$.
For every integer $s\ge k_j$, the sum bound in that equation gives
\eq
\norm{K_{\beta_j,D}^{\,s}-\lambda_j^sP_j}_1
 =\sum_{\xi_a\ne\lambda_j}|\xi_a|^s
 \le\lambda_j^s q^{s-k_j}d(u_j)
 \le\frac1{10}\lambda_j^s q^s.
\lb{residualall}
\en
Here the eigenvalues are counted with multiplicity, and $\norm{\cdot}_1$
is the trace norm, which for a self-adjoint trace-class operator is the
sum of the absolute values of its eigenvalues.
The same individual-eigenvalue bound gives, for every integer $s\ge1$,
\eq
\left\|K_{\beta_j,D}^{s}-\lambda_j^sP_j\right\|
 \le\lambda_j^s q^s.
\lb{operatorconcentration}
\en
Unlike the trace-norm bound, this operator-norm estimate does not
require $s\ge k_j$. Neither estimate requires $K_{\beta_j,D}$ itself
to be positive.

Let $A$ be supported on an interval of $m$ sites.
The indicator of the empty history is a unit vector in $\calK_\beta$.
Its Rayleigh quotient, together with positivity of the largest
transfer eigenvalue at $\beta_j$, gives
\eq
\lambda_j\ge\tr e^{-\beta_j DQ}
 \ge3e^{-\beta_j\max(D,0)}.
\lb{lambdalower}
\en
Combining this with \rlb{Bnorm}, we obtain the estimate
\eq
\frac{\norm{\mathcal B_{A,\beta_j,D}}}{\lambda_j^m}
 \le\norm A\,e^{(3+|D|)m\beta_j}
 \le\norm A\,e^{(3+D_0)m\beta_j}.
\lb{Bnormalizednorm}
\en
For fixed support, the quotient grows at most exponentially in
$\beta_j$.

Recall the exact finite-chain formula and its fixed-temperature limit:
\eqa
\rho_{N,\beta_j,D}^{\per}(A)
 &=\frac{\Tr\bigl(\mathcal B_{A,\beta_j,D}
                   K_{\beta_j,D}^{\,N-m}\bigr)}
        {Z_{N,D}(\beta_j)},
 \lb{perexpectation}\\
\nu_{\beta_j,D}(A)
 &=\frac{\Tr\bigl(P_j\mathcal B_{A,\beta_j,D}\bigr)}{\lambda_j^m}.
 \lb{nudef}
\ena
The second expression defines the state obtained in
\rlb{nufixedbeta}, so $|\nu_{\beta_j,D}(A)|\le\norm A$.
Subtract the leading term from the transfer power in the numerator
of \rlb{perexpectation}. Since
$\lambda_j^{N-m}\Tr(P_j\mathcal B_{A,\beta_j,D})
=\lambda_j^N\nu_{\beta_j,D}(A)$, the exact difference is
\eq
\begin{split}
\rho_{N,\beta_j,D}^{\per}(A)-\nu_{\beta_j,D}(A)
 ={}&\frac{\Tr\!\left[\mathcal B_{A,\beta_j,D}
       (K_{\beta_j,D}^{\,N-m}-\lambda_j^{N-m}P_j)\right]}
             {Z_{N,D}(\beta_j)}\\
 &-\nu_{\beta_j,D}(A)
       \frac{Z_{N,D}(\beta_j)-\lambda_j^N}{Z_{N,D}(\beta_j)}.
\end{split}
\lb{perresidualidentity}
\en
This identity holds at either parity.
For the estimate, take $N$ even, so that
$Z_{N,D}(\beta_j)\ge\lambda_j^N$.
If $N-m\ge k_j$, \rlb{residualall} bounds the first numerator by
$\norm{\mathcal B_{A,\beta_j,D}}\lambda_j^{N-m}q^{N-m}/10$
and the absolute value of the second numerator by
$\lambda_j^Nq^N/10$.
Using \rlb{Bnormalizednorm} gives
\eq
\abs{\rho_{N,\beta_j,D}^{\per}(A)-\nu_{\beta_j,D}(A)}
 \le\frac{\norm A}{10}
       \left\{e^{(3+D_0)m\beta_j}q^{N-m}+q^N\right\}.
\lb{exactpererror}
\en
The factor $\lambda_j^{-m}$ is accounted for by the $m$ transfer
factors replaced by the insertion block.
No eigenspace was discarded in the finite-chain identity; the
contributions other than the leading one have been bounded.
The periodic Gibbs expectation is therefore close to
$\nu_{\beta_j,D}(A)$ once the sites outside the support of $A$ are
sufficiently numerous. These bounds also control repeated insertions separated by long
intervals, which will be used in the uniqueness proof.

\subsection{The finite periodic gap}

We recall the elementary implication of physical purity. If the
Gibbs probabilities of a finite Hamiltonian, listed with multiplicity,
have purity at least $1-\varepsilon>1/2$, then their largest value
$p_0$ is at least the purity. Thus the ground eigenvector is unique,
and the probability of an individual first excited eigenvector obeys
\eq
e^{-\beta\gamma}=\frac{p_1}{p_0}
 \le\frac{1-p_0}{p_0}\le\frac{\varepsilon}{1-\varepsilon}.
\lb{puritygapbasic}
\en
This uses individual eigenvector probabilities, not the total weight
of an energy eigenspace.

\begin{lemma}[Periodic gap]\label{L:periodicgap}
Under \rlb{periodiccriterion}, $H_N^{\per}(D)$ has a unique ground
state with gap greater than $\gamma_0$ at every even $N\ge k_0$,
uniformly for $|D|\le D_0$.
\end{lemma}
\begin{proof}
For a nonnegative list and exponents $b\ge a>0$,
$\sum_i z_i^b\le(\sum_i z_i^a)^{b/a}$.
Apply this to the absolute transfer eigenvalues at the two temperatures.
For even $k\le N\le2k$, it gives
$Z_{N,D}(\beta)\le Z_{k,D}(\beta)^{N/k}$ and
$Z_{N,D}(2\beta)\ge Z_{2k,D}(2\beta)^{N/(2k)}$.
Dividing and using the definitions yields
\eq
T_D(N,\beta_j)
 \ge\{T_D(2k_j,\beta_j)S_D(k_j,\beta_j)^2\}^{N/(2k_j)}
 \ge(1-u_j)^3\ge1-3u_j
\lb{interpolation}
\en
when $k_j\le N\le2k_j$. The bracket is in $(0,1]$ and the
exponent is at most one, which explains the second inequality.
Equation \rlb{puritygapbasic} now gives
\eq
e^{-\beta_j\gamma_N^{\per}(D)}
 \le\frac{3u_j}{1-3u_j}
 <6u_j=(6\eta)^{2^j}<(10\eta)^{2^j}
 =e^{-\gamma_0\beta_j}.
\lb{pergap}
\en
The strict first comparison uses $u_j\le1/100$; the identities
follow from \rlb{errorscales}, \rlb{gamma0}, and \rlb{scales}.
The intervals $[k_j,2k_j]$ cover all lengths $N\ge k_0$.
\end{proof}

We have obtained a simple positive leading transfer eigenvalue,
uniform spatial contraction, control between successive temperatures,
and finite periodic gaps from \rlb{periodiccriterion} alone.
The next section uses this information to identify every
translation-invariant ground state, not just limits of periodic ones.

\section{Uniqueness among translation-invariant ground states}
\label{S:uniqueness}

We now prove uniqueness among translation-invariant ground states. Only the periodic criterion
\rlb{periodiccriterion} is assumed in this section. In particular, no
endpoint vector, boundary-field estimate, or convergence at zero
temperature is used as an input. We first identify the energy density
of an arbitrary translation-invariant ground state. We then compare
its exponential moments with those of the periodic Gibbs states.

\subsection{The energy density of a translation-invariant ground state}

Recall the open-chain Hamiltonian $H_N^{\free}(D)$ in
\rlb{Hfree} and its ground-state energy
$E_N^{\free}(D)=\min\operatorname{spec}H_N^{\free}(D)$.
Deleting or adding one bond changes the ground energy by at most two.
Consequently, for $n,m\ge2$,
\eq
|E_{n+m}^{\free}(D)-E_n^{\free}(D)-E_m^{\free}(D)|\le2.
\lb{almostadditive}
\en
Subadditivity of $E_N^{\free}(D)+2$ and superadditivity of $E_N^{\free}(D)-2$
show that the limit
\eq
e(D)=\lim_{N\uparrow\infty}\frac{E_N^{\free}(D)}N
     =\lim_{L\uparrow\infty}\frac{E_{2L}^{\per}(D)}{2L}
\lb{energydensity}
\en
exists. Here $E_N^{\per}(D)$ denotes the lowest energy of
$H_N^{\per}(D)$; the periodic limit is taken only through even
lengths. Iterating
\rlb{almostadditive} at a fixed length $N$ and letting the number of
copies tend to infinity gives
\eq
|E_N^{\free}(D)-Ne(D)|\le2,
\qquad |E_{2L}^{\per}(D)-2Le(D)|\le4.
\lb{energyfiniteerror}
\en

\begin{lemma}[Minimum energy density]\label{L:energydensity}
Let $\theta$ be any translation-invariant ground state of $H(D)$.
Then
\eq
\theta(\spn_0\cdot\spn_1+DQ_0)=e(D).
\lb{thetameanenergy}
\en
No purity, spectral-gap, or spin-symmetry assumption is imposed on
$\theta$.
\end{lemma}
\begin{proof}
The assertion is a direct consequence of
\cite[Theorem~6.2.58]{BR2}, which characterizes translation-invariant
ground states as minimizers of the mean energy among
translation-invariant states; see also \cite{BKR}.
Its hypotheses are satisfied by our finite-range interaction.
The minimum mean energy equals $e(D)$ in \rlb{energydensity},
by the finite-chain variational principle and the construction of
limits of periodic-chain ground states following
Theorem~\ref{T:periodic}.
For completeness, we give a direct proof in the present setting.

Define the interaction joining $\Lambda_N$ to its complement by
\eq
\mathit{\Delta}H_N
=\spn_{\ell_N-1}\cdot\spn_{\ell_N}
 +\spn_{r_N}\cdot\spn_{r_N+1}.
\lb{boundaryinteraction}
\en
Since each bond operator has norm two,
$\norm{\mathit{\Delta}H_N}\le4$.
The single-ion terms do not cross the boundary.
Choose a normalized ground vector
$\psi_N$ of $H_N^{\free}(D)$ and an orthonormal basis $(e_a)$ of its physical
Hilbert space. The local operators
\eq
C_a=|\psi_N\rangle\langle e_a|,
\qquad \sum_a C_a^*C_a=\one,
\lb{replacementoperators}
\en
define the unital completely positive map
$\mathcal T_N(B)=\sum_a C_a^*BC_a$ on the quasi-local algebra.
It replaces the state inside $\Lambda_N$ by the vector state of
$\psi_N$ and leaves the exterior algebra fixed. In particular,
$\mathcal T_N(H_N^{\free}(D))=E_N^{\free}(D)\one$ and
$\norm{\mathcal T_N(\mathit{\Delta}H_N)}\le4$.
Summing the ground-state inequalities for the $C_a$ gives
\eq
\begin{split}
0&\le\sum_a\theta\bigl(C_a^*[H(D),C_a]\bigr)\\
 &=\theta\bigl(\mathcal T_N(H_N^{\free}(D)+\mathit{\Delta}H_N)
                    -H_N^{\free}(D)-\mathit{\Delta}H_N\bigr).
\end{split}
\lb{replacementenergy}
\en
The commutators are finite sums; all other interaction terms commute
with each $C_a$. It follows that
\eq
E_N^{\free}(D)\le\theta(H_N^{\free}(D))\le E_N^{\free}(D)+8.
\lb{localenergycomparison}
\en
Writing $e_\theta=\theta(\spn_0\cdot\spn_1+DQ_0)$, translation
invariance gives
$\theta(H_N^{\free}(D))=Ne_\theta-\theta(\spn_0\cdot\spn_1)$.
Divide \rlb{localenergycomparison} by $N$ and let $N\uparrow\infty$.
This proves $e_\theta=e(D)$.
\end{proof}

The compactness argument following Theorem~\ref{T:periodic},
using the classical finite-chain uniqueness result, already
establishes the existence of a translation-invariant ground state.
That existence argument does not use the uniqueness theorem for
the infinite chain or the gap estimates.
Lemma~\ref{L:energydensity} identifies the energy density of every
state to which the uniqueness theorem applies. We next compare
this density with the free energy at low temperature.

\subsection{The free-energy correction at low temperature}

At a selected inverse temperature $\beta_j$, the free energy per
site is obtained from the simple leading eigenvalue. By
\rlb{residualall},
\eq
\lim_{\substack{N\uparrow\infty\\N\ {\rm even}}}
 \left(-\frac{1}{\beta_jN}\log Z_{N,D}(\beta_j)\right)
 =-\frac{\log\lambda_j}{\beta_j}.
\lb{freeenergylimit}
\en
The elementary finite-chain bounds
$e^{-\beta E_N^{\per}(D)}\le Z_{N,D}(\beta)
\le3^N e^{-\beta E_N^{\per}(D)}$, together with
\rlb{energydensity}, imply
\eq
0\le\log\lambda_j+\beta_j e(D)\le\log3,
\qquad
-\lim_{j\uparrow\infty}\frac{\log\lambda_j}{\beta_j}=e(D).
\lb{freeenergyelementary}
\en
The nonnegative quantity $\log\lambda_j+\beta_j e(D)$ is
$\beta_j$ times the difference between the ground-state energy
density and the free-energy density. The following sharper bound
is the first new use of the comparison between successive leading
eigenvalues.

\begin{lemma}[A small free-energy correction]\label{L:freeenergy}
Under \rlb{periodiccriterion}, for every $j\ge0$,
\eq
0\le\log\lambda_j+\beta_j e(D)\le\frac{u_j}{k_j}.
\lb{freeenergysmall}
\en
The bound is uniform for $|D|\le D_0$.
\end{lemma}
\begin{proof}
By the definition \rlb{rdefinition} and
$\beta_{j+1}=2\beta_j$ in \rlb{scales},
\eq
\frac{\log\lambda_{j+1}}{\beta_{j+1}}
 -\frac{\log\lambda_j}{\beta_j}
 =\frac{\log r_j}{2\beta_j}.
\lb{freeenergydifference}
\en
Telescoping this equality and using the limit in
\rlb{freeenergyelementary} gives the exact identity
\eq
\log\lambda_j+\beta_j e(D)
 =-\sum_{\ell=j}^{\infty}2^{j-\ell-1}\log r_\ell.
\lb{freeenergytelescoping}
\en
The lower bound in \rlb{rbound} and
$-\log(1-u)\le u/(1-u)$ imply
\eq
0\le-\log r_\ell
 \le\frac{u_\ell+u_{\ell+1}}{2k_\ell(1-\eta)}
 \le\frac{u_\ell}{k_\ell(1-\eta)}.
\lb{logrfreeenergy}
\en
Here $u_{\ell+1}\le u_\ell\le\eta$. Substitution into
\rlb{freeenergytelescoping}, followed by $u_{j+a}\le u_j$ and
$k_{j+a}=2^ak_j$, yields
\eq
\log\lambda_j+\beta_j e(D)
 \le\frac{u_j}{2k_j(1-\eta)}\sum_{a=0}^{\infty}4^{-a}
 =\frac{2u_j}{3k_j(1-\eta)}\le\frac{u_j}{k_j}.
\lb{freeenergygeometric}
\en
The final inequality uses $\eta\le1/100$.
\end{proof}

Unlike the elementary bound by $\log3$, \rlb{freeenergysmall} remains
small after multiplication by a length of order $k_j^2$, because
$k_j^2u_j/k_j=k_ju_j\longrightarrow0$ by \rlb{errorscales}.
This is essential for the uniqueness proof: intervals of this length
will suppress the contribution of every other transfer eigenvalue,
while the free-energy contribution associated with each interval
still tends to zero.

\subsection{Exponentials of separated local observables}
\label{S:separated}

Fix a self-adjoint operator $A$ supported on $m$ consecutive sites.
For this calculation only, label the periodic ring by
$0,\ldots,N-1$ and identify the original block with $0,\ldots,m-1$.
Let us define a translated sum of $A$ by
\eq
s_j=m+2k_j^2,\qquad
A^{(n)}=\sum_{r=0}^{n-1}\tau_{rs_j}(A),\qquad N=ns_j,
\lb{separatedblocks}
\en
where $n\ge2$ is even, so $N$ is even whatever the parity of $m$.
The blocks are disjoint. Hence their observables commute and
$e^{tA^{(n)}}=\prod_{r=0}^{n-1}\tau_{rs_j}(e^{tA})$ for real $t$.

Applying the same bond expansion as in
Section~\ref{S:blockconstruction} to every block gives the exact
multiple-insertion identity
\eq
\Tr_{\mathcal H_N}\bigl(e^{tA^{(n)}}e^{-\beta_jH_N^{\per}(D)}\bigr)
 =\Tr_{\calK_{\beta_j}}
       \left(\mathcal B_{e^{tA},\beta_j,D}K_{\beta_j,D}^{2k_j^2}\right)^n.
\lb{multipleinsertion}
\en
The expansion is absolutely convergent for each finite $n$: its
absolute sum is bounded by
$3^N e^{|t|n\norm A+(3+D_-)N\beta_j}$.
There are $m$ site kernels in each insertion and $2k_j^2$ untouched
kernels between consecutive blocks, including the last and first.
No commutation between an insertion and $K$ is used.

To bound the trace, introduce the operator associated with one
insertion block and $k_j^2$ sites on each of its two sides:
\eq
\mathcal M_j(t)
 =\frac{K_{\beta_j,D}^{k_j^2}\,
             \mathcal B_{e^{tA},\beta_j,D}\,
             K_{\beta_j,D}^{k_j^2}}
        {\lambda_j^{m+2k_j^2}}.
\lb{Mblock}
\en
This notation is used only for the repeated-insertion calculation.
The operator is Hilbert--Schmidt, and cyclicity of products of
Hilbert--Schmidt operators gives
\eq
\Tr\left(\mathcal B_{e^{tA},\beta_j,D}K_{\beta_j,D}^{2k_j^2}\right)^n
 =\lambda_j^{ns_j}\Tr\mathcal M_j(t)^n.
\lb{Mtraceidentity}
\en
The rank-one identity $P_j\mathcal B P_j=\Tr(P_j\mathcal B)P_j$
and \rlb{operatorconcentration} imply
\eq
\left\|\mathcal M_j(t)-\nu_{\beta_j,D}(e^{tA})P_j\right\|
 \le3e^{|t|\norm A+(3+D_0)m\beta_j}q^{k_j^2}.
\lb{Mleading}
\en
To see the bound explicitly, replace each of the two powers
$K_{\beta_j,D}^{k_j^2}$ in \rlb{Mblock} by
$\lambda_j^{k_j^2}P_j$ plus its remainder. The two linear remainder
terms contribute at most $2q^{k_j^2}$ times
$\norm{\mathcal B_{e^{tA},\beta_j,D}}/\lambda_j^m$, and the
quadratic term contributes at most $q^{2k_j^2}$ times that quotient.
Now use \rlb{Bnormalizednorm} and $\norm{e^{tA}}\le e^{|t|\norm A}$.

\begin{lemma}[Separated exponential moment]\label{L:exponentialmoment}
Under \rlb{periodiccriterion}, for $A=A^*$ and real $t$,
\eq
\begin{split}
\limsup_{\substack{n\uparrow\infty\\n\ {\rm even}}}
 \frac1n\log\rho_{ns_j,\beta_j,D}^{\per}(e^{tA^{(n)}})
 \le\log\left\{\nu_{\beta_j,D}(e^{tA})
      +3e^{|t|\norm A+(3+D_0)m\beta_j}q^{k_j^2}\right\}.
\end{split}
\lb{exponentialmoment}
\en
\end{lemma}
\begin{proof}
At even $N$, $Z_{N,D}(\beta_j)\ge\lambda_j^N$.
The left side of \rlb{multipleinsertion} is positive, since it is
the trace of the product of two positive definite physical matrices.
For $n\ge2$, the trace-ideal inequalities give
\eq
|\Tr\mathcal M_j(t)^n|
 \le\norm{\mathcal M_j(t)}^{n-2}
                         \norm{\mathcal M_j(t)}_{\rm HS}^{\,2}.
\lb{powertracebound}
\en
This does not require $\mathcal M_j(t)$ to be normal, self-adjoint,
or positive. Divide \rlb{multipleinsertion} by $Z_{N,D}(\beta_j)$,
use \rlb{Mtraceidentity}, and take logarithms divided by $n$.
At fixed $j,t,A$, the Hilbert--Schmidt factor in
\rlb{powertracebound} disappears as $n\uparrow\infty$.
Finally apply \rlb{Mleading}. The argument of the logarithm is
positive because $\nu_{\beta_j,D}(e^{tA})\ge e^{-|t|\norm A}>0$.
\end{proof}

The order of limits matters: first the number $n$ of separated blocks
tends to infinity at fixed $j$, and only later does $j$ tend to
infinity. We have not asserted factorization of a many-point
expectation with an error uniform in $n$; instead we have bounded
its logarithm per block.

\subsection{A trace comparison and uniqueness}

\begin{proposition}[Uniqueness from the periodic criterion]
\label{P:TIunique}
Under \rlb{periodiccriterion}, there is exactly one
translation-invariant ground state $\omega_D$ for every $|D|\le D_0$.
Moreover $\nu_{\beta_j,D}(A)\to\omega_D(A)$ as $j\uparrow\infty$
for every local operator $A$.
\end{proposition}
\begin{proof}
Let $\theta$ be any translation-invariant ground state. By
Lemma~\ref{L:energydensity},
\eq
\bigl|\theta(H_N^{\per}(D))-Ne(D)\bigr|\le4.
\lb{thetaperenergy}
\en
To see this, note that the $N-1$ internal bonds and all $N$ single-ion terms have
expectation $Ne(D)-\theta(\spn_0\cdot\spn_1)$; the closing bond
has norm two. This estimate does not assume clustering of $\theta$.

For the separated observable in \rlb{separatedblocks}, scalar
Jensen's inequality and the Golden--Thompson inequality
\cite{Golden,Thompson} give
\eq
\begin{split}
\exp\{tn\theta(A)-\beta_j\theta(H_N^{\per}(D))\}
 &\le\theta\left(e^{-\beta_jH_N^{\per}(D)+tA^{(n)}}\right)\\
 &\le\Tr_{\mathcal H_N}e^{-\beta_jH_N^{\per}(D)+tA^{(n)}}\\
 &\le Z_{N,D}(\beta_j)\,
             \rho_{N,\beta_j,D}^{\per}(e^{tA^{(n)}}).
\end{split}
\lb{JensenGT}
\en
In the first line, Jensen is applied to the scalar spectral measure
of the self-adjoint matrix $-\beta_jH_N^{\per}(D)+tA^{(n)}$ in
the restricted state $\theta$. Translation invariance gives
$\theta(A^{(n)})=n\theta(A)$. For the second line, a density matrix
has eigenvalues at most one, so its expectation of a positive matrix
is at most the unnormalized trace.\footnote{This is a very crude
comparison: for the identity matrix it says only $1\le3^N$.
Nevertheless it suffices here. After taking logarithms and dividing
by the number of separated blocks, the unperturbed contribution is
$s_j\{\log\lambda_j+\beta_j e(D)\}$, which tends to zero by
Lemma~\ref{L:freeenergy}. No estimate of the competing state's
entropy is needed.} The last line is
Golden--Thompson followed by cyclicity of the physical trace.
Thus no operator-convexity assertion about the exponential and no
relative-entropy estimate is needed.

Take logarithms in \rlb{JensenGT}, divide by $n$, and let even
$n\uparrow\infty$. Equations \rlb{freeenergylimit},
\rlb{thetaperenergy}, and \rlb{exponentialmoment} yield
\eq
\begin{split}
t\theta(A)\le{}&s_j\{\log\lambda_j+\beta_j e(D)\}\\
 &+\log\left\{\nu_{\beta_j,D}(e^{tA})
       +3e^{|t|\norm A+(3+D_0)m\beta_j}q^{k_j^2}\right\}.
\end{split}
\lb{masteruniqueness}
\en
The first term tends to zero as $j\uparrow\infty$, because
\rlb{freeenergysmall} bounds it
by $(m+2k_j^2)u_j/k_j$. The additive error inside the logarithm
also tends to zero: its negative exponent is quadratic in $k_j$,
whereas $\beta_j=(\beta_0/k_0)k_j$ grows only linearly.

Choose any state $\nu_*$ obtained as a local limit of a subsequence
$\nu_{\beta_{j_a},D}$, where $j_1<j_2<\cdots$. Such a subsequence
exists by compactness of the state space of each finite local
algebra and a diagonal argument. Thus
$\nu_{\beta_{j_a},D}(B)\to\nu_*(B)$ for every local operator $B$.
Taking $j=j_a$ and then $a\uparrow\infty$ in
\rlb{masteruniqueness} gives
\eq
t\theta(A)\le\log\nu_*(e^{tA}),\qquad t\in\bbR.
\lb{limitJensen}
\en
The right side vanishes at $t=0$ and has derivative $\nu_*(A)$
there. Dividing by $t>0$ and letting $t\downarrow0$ gives
$\theta(A)\le\nu_*(A)$. Dividing by $t<0$ and letting
$t\uparrow0$ reverses the inequality and gives
$\theta(A)\ge\nu_*(A)$. Hence $\theta(A)=\nu_*(A)$.
This holds for every self-adjoint local $A$, and therefore for every
local operator by linearity.

There exists at least one translation-invariant ground state, as
observed above. Since every such state equals every subsequential limit
$\nu_*$, there is exactly one; denote it by $\omega_D$.
The same argument identifies the limit of every convergent
subsequence of $\nu_{\beta_j,D}$, and hence proves
\eq
\lim_{j\uparrow\infty}\nu_{\beta_j,D}(A)=\omega_D(A).
\lb{nulimit}
\en
No convergence of these states, nor a gap for $\theta$, was assumed
in proving uniqueness. This completes the proof of translation-invariant uniqueness.
\end{proof}

\subsection{A quantitative local estimate}
\label{S:quantitativenulimit}

The same trace comparison gives a rate, which will be useful for
finite-chain convergence. If $A=A^*$ and $\norm A\le1$, the
scalar spectral measure of $A$ in any state $\nu$ is supported on
$[-1,1]$. The second derivative of $\log\nu(e^{tA})$ is the variance
in the exponentially tilted measure and is at most one. Therefore
\eq
\log\nu(e^{tA})\le t\nu(A)+\frac{t^2}{2},\qquad t\in\bbR.
\lb{scalarlogmoment}
\en
For $|t|\le1$, also $\nu(e^{tA})\ge e^{-1}$ and
$\log(a+b)\le\log a+b/a$ for $a>0$, $b\ge0$.
Apply these facts to \rlb{masteruniqueness} with $\theta=\omega_D$,
and then apply the result to both signs of $A$. For $0<t\le1$,
\eq
\begin{split}
|\omega_D(A)-\nu_{\beta_j,D}(A)|
 \le\frac t2+\frac1t\left\{
       \frac{(m+2k_j^2)u_j}{k_j}
       +3e^{2+(3+D_0)m\beta_j}q^{k_j^2}\right\}.
\end{split}
\lb{nulimitoptimized}
\en
The quantity in braces tends to zero uniformly in $D$ and in such
$A$. For large $j$, choose $t$ to be the square root of twice that
quantity; it is then at most one. Increasing a constant to cover the
finitely many other $j$, and decomposing a general $A$ into its
self-adjoint real and imaginary parts, gives
\eq
\begin{split}
|\omega_D(A)-\nu_{\beta_j,D}(A)|
 \le C_I\norm A\left\{\sqrt{k_ju_j}
       +e^{(3+D_0)m\beta_j/2}q^{k_j^2/2}\right\},
 \qquad A\in\calA_I.
\end{split}
\lb{nurate}
\en
The constant $C_I$ depends on the fixed support, not on $D$ or $j$.

\subsection{Continuity, symmetry, and the local gap}
\label{S:bulkconsequences}

Proposition~\ref{P:TIunique} and the argument following
Theorem~\ref{T:periodic} prove the periodic limit \rlb{PBClimit}.
We next prove continuity of $\omega_D$ in $D$.
Let $(D_a)_{a=1,2,\ldots}$ be any sequence in $[-D_0,D_0]$ with
$D_a\to D$. A subsequence of the states $\omega_{D_a}$ has a
local limit $\theta$. This state is translation invariant. For
local $A$, the norm convergence
$[H(D_a),A]\to[H(D),A]$ allows the ground-state inequalities to
pass to the limit, so $\theta$ is a ground state of $H(D)$.
Uniqueness at $D$ gives $\theta=\omega_D$. The same conclusion
holds for every convergent subsequence, proving
$\omega_{D_a}(A)\to\omega_D(A)$ for every local $A$.
Norm approximation then gives continuity for every $A\in\calA$.

The periodic convergence is also uniform for $|D|\le D_0$.
Otherwise, for some fixed local observable, one could choose
$L_a\uparrow\infty$ and $D_a\to D$ along a subsequence on which
the claimed convergence fails. A further local limit of the
periodic ground-state expectations is a translation-invariant
ground state of $H(D)$ and must equal $\omega_D$. The continuity
just proved gives the same limit for $\omega_{D_a}$, a
contradiction. Since each fixed local algebra is finite dimensional,
this convergence is uniform also on its unit ball.

Every finite periodic ground-state expectation has the symmetries of
its Hamiltonian, since its ground eigenspace is one dimensional.
Their limits inherit reflections, uniform rotations about the
$\mathrm z$ axis, and the uniform $\pi$ rotation about the
$\mathrm x$ axis; at $D=0$ they inherit all spin rotations.
For a local operator $A$, the finite-chain gap gives
\eq
\omega_{N,D}^{\per}\bigl(A^*[H_N^{\per}(D),A]\bigr)
 \ge\gamma_0\left\{\omega_{N,D}^{\per}(A^*A)
                   -|\omega_{N,D}^{\per}(A)|^2\right\}.
\lb{finitegapineq}
\en
For sufficiently large centered rings, the commutator equals
$[H(D),A]$. Passing to \rlb{PBClimit} proves
\rlb{localgap}. The strict finite-chain gap becomes a non-strict
lower bound in the limit. A locally-unique gapped ground state is
pure by \cite[Theorem~2.4 and Appendix~A]{TasakiLSMReview}.

This proves Theorem~\ref{T:main} under the periodic criterion,
as well as the qualitative convergence and uniformity in
Theorem~\ref{T:periodic}. Its finite-chain gap assertion was proved
in Lemma~\ref{L:periodicgap}; its quantitative convergence bound
will follow in Section~\ref{S:quantitativerate} from the periodic
estimates. Appendix~\ref{S:numerics} supplies explicit parameters.
In particular, the infinite-chain state is determined before any
endpoint-field estimates are introduced. The next sections show
that those fields select the same state.
\section{Open chains with endpoint magnetic fields}
\label{S:fieldtransfer}

The unique translation-invariant ground state $\omega_D$ is now determined. We return to the
finite chains with equal endpoint fields and show that their local
limits are translation invariant. The exact transfer construction
uses the same $K_{\beta,D}$ and the same insertion block as before;
only the endpoint contractions change.

\subsection{Open chains with equal endpoint fields}
\label{S:equalfields}

We next derive the representation for the open chain with the boundary condition of
interest.
Write
\eq
Y_{N,D}^{\fld}(\beta)=\Tr_{\mathcal H_N}e^{-\beta H_N^{\fld}(D)},
\qquad N\ge2.
\lb{Yfielddefinition}
\en
A field term $-h_0\Sz$ contributes $+h_0\Sz=-ih_0G^{\mathrm z}$
to the negative Hamiltonian in the ordered expansion.
We can therefore represent its events by an extra history at the
endpoint: only $\mathrm z$ marks occur, and each carries the scalar
coefficient $-ih_0$. These endpoint histories are integrated just like
bond histories, but do not add physical sites or bonds.

For $\alpha\in\{\mathrm x,\mathrm y,\mathrm z\}$, let $r_\alpha(x)$
be the number of marks of color $\alpha$ in $x\in\Omega_\beta$.
Define the boundary vector
\eq
v_\beta(x)=
\begin{cases}
 (-ih_0)^{r_{\mathrm z}(x)},&r_{\mathrm x}(x)=r_{\mathrm y}(x)=0,\\
 0,&\text{otherwise}.
\end{cases}
\lb{boundaryvector}
\en
The empty-history value is one. Note also that
\eq
\norm{v_\beta}^2
 =\int_{\Omega_\beta}d\mu_\beta(x)\,|v_\beta(x)|^2
 =\sum_{r=0}^{\infty}\frac{\beta^rh_0^{2r}}{r!}
 =e^{\beta h_0^2}.
\lb{boundaryvectornorm}
\en
The factors $\beta^r/r!$ come from
the ordered time integrations, not from the definition of $v_\beta$.
The vector is independent of $D$: the single-ion term is already
included in $k_{\beta,D}$.

Our inner product is antilinear in its first argument. Thus $v_\beta$
already supplies the desired coefficient on the \emph{right}, whereas
on the left we need its complex conjugate as a vector.
Introduce
\eq
(\Ux f)(x)=(-1)^{r_{\mathrm y}(x)+r_{\mathrm z}(x)}f(x).
\lb{Uxdefinition}
\en
This is the self-adjoint unitary representing a $\pi$ rotation about
the $\mathrm x$ axis in history space. In particular,
\eq
\Ux^*=\Ux,\qquad \Ux^2=\one,\qquad
\Ux v_\beta=\overline{v_\beta},\qquad
\Ux K_{\beta,D}=K_{\beta,D}\Ux.
\lb{Uxcommutes}
\en
For the last identity, rotate every matrix in the single-site word:
$G^{\mathrm x}$ is unchanged, $G^{\mathrm y}$ and $G^{\mathrm z}$
change sign, and $Q$ is unchanged. Thus a nonzero kernel can only
connect histories with the same $\Ux$ sign, which proves commutation.

At each end, the coefficient entering the expansion is now
$(-ih_0)^r$: on the right it is $v_\beta$, and on the left it is
$\overline{\Ux v_\beta}=v_\beta$.
Regrouping the $N$ site words, with the $N-1$ internal bond histories
and the two endpoint histories, yields
\eq
Y_{N,D}^{\fld}(\beta)
 =\langle\Ux v_\beta,K_{\beta,D}^{N}v_\beta\rangle.
\lb{Yfieldtransfer}
\en
There are $N$ kernels because there are $N$ physical sites.
The endpoint histories are already accounted for by the two vectors.
The absolute sum of the traced terms is bounded by
$3^N\exp\{\beta[ND_-+3(N-1)+2h_0]\}$, which justifies the
regrouping. See Figure~\ref{fig:field}(a).

\begin{figure}[tbp]
\centering
\includegraphics[width=0.85\textwidth]{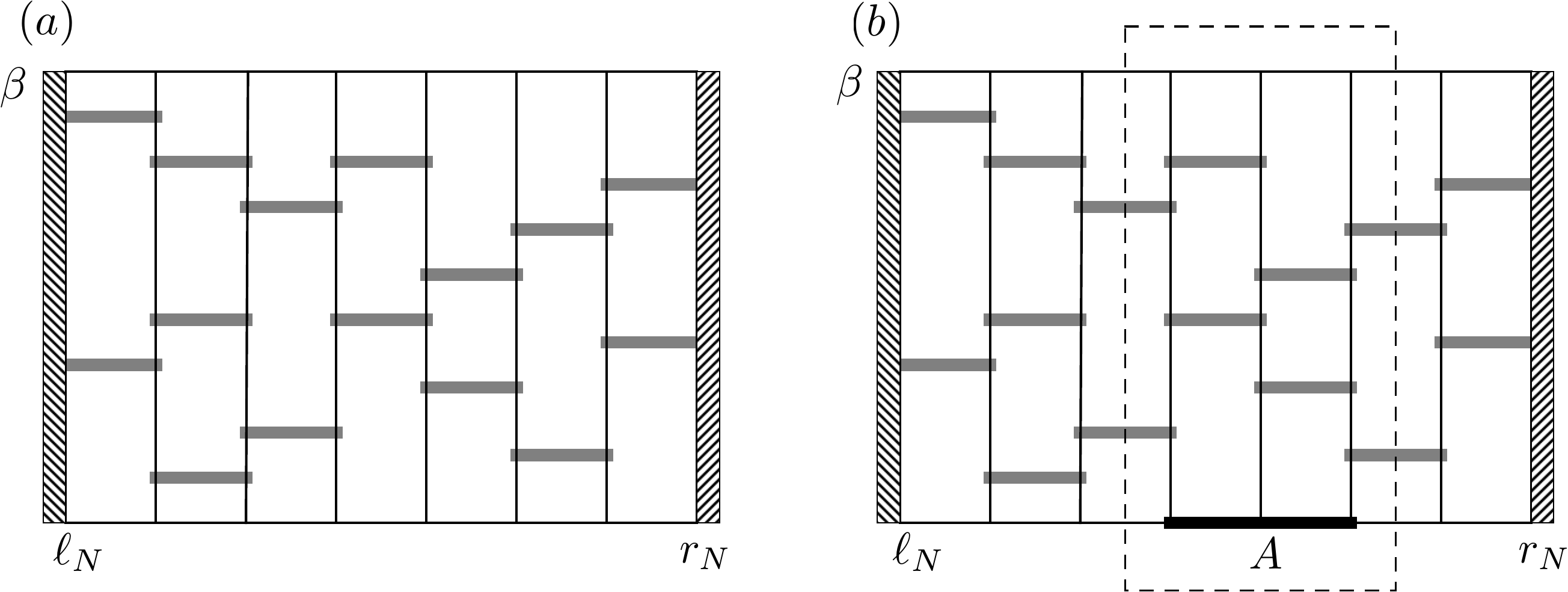}
\caption{The open chain with equal endpoint fields. Panel (a) represents
$Y_{N,D}^{\fld}(\beta)$ and panel (b) represents
$\Tr_{\mathcal H_N}(Ae^{-\beta H_N^{\fld}(D)})$.
The hatched strips indicate the two equal field terms $-h_0\Sz$;
their ordered field events are integrated against $\Ux v_\beta$ on
the left and $v_\beta$ on the right, with conjugation in the left slot
of the inner product. There is no spatial identification of the two
ends. The insertion block in (b) is exactly the same block operator
$\mathcal B_{A,\beta,D}$ as in Figure~\ref{fig:periodic}; only the
way its exterior site kernels are contracted has changed. In panel (b), the block contains $m=3$ sites,
with $d_{\rm L}=3$ sites to its left and $d_{\rm R}=2$ to its right,
so the drawing has $N=8$. The endpoint labels $\ell_N,r_N$ cover both
parities; the endpoint-field ground-state construction uses odd
lengths.}
\label{fig:field}
\end{figure}

For $I=[s,t]\subset\Lambda_N=\{\ell_N,\ldots,r_N\}$, put
\eq
d_{\rm L}=s-\ell_N,\qquad d_{\rm R}=r_N-t,\qquad
m=t-s+1,\qquad N=d_{\rm L}+m+d_{\rm R}.
\lb{distances}
\en
Inserting $A$ at time $0$ as in Figure~\ref{fig:field}(b) gives
\eq
\Tr_{\mathcal H_N}\bigl(Ae^{-\beta H_N^{\fld}(D)}\bigr)
 =\langle\Ux v_\beta,
   K_{\beta,D}^{d_{\rm L}}\mathcal B_{A,\beta,D}
   K_{\beta,D}^{d_{\rm R}}v_\beta\rangle.
\lb{fieldunnormalized}
\en
Consequently,
\eq
\rho_{N,\beta,D}^{\fld}(A)
 =\frac{\langle\Ux v_\beta,
   K_{\beta,D}^{d_{\rm L}}\mathcal B_{A,\beta,D}
   K_{\beta,D}^{d_{\rm R}}v_\beta\rangle}
        {\langle\Ux v_\beta,K_{\beta,D}^{N}v_\beta\rangle}.
\lb{fieldGibbsraw}
\en
This is again an exact identity at both parities. In particular, the
same insertion operator describes the periodic and open expectations.
We retain $\Ux$ on the left throughout; it cannot in general be moved
across $\mathcal B_A$ without also rotating the observable.

\subsection{The auxiliary opposite-field chain}
\label{S:oppositefields}

If the same vector $v_\beta$ is used in both slots of \rlb{Yfieldtransfer}, conjugation reverses
the sign of the left field. This motivates us to define
\eq
H_N^{\opp}(D)
 =\sum_{j=\ell_N}^{r_N-1}\spn_j\cdot\spn_{j+1}
  +D\sum_{j\in\Lambda_N}Q_j
  +h_0\Sz_{\ell_N}-h_0\Sz_{r_N},
\lb{Hopp}
\en
and, for every $N\ge2$, set
\eq
Y_{N,D}(\beta)=\Tr_{\mathcal H_N}e^{-\beta H_N^{\opp}(D)}
 =\langle v_\beta,K_{\beta,D}^{N}v_\beta\rangle.
\lb{Ytransfer}
\en
Thus $Z$ always means periodic boundary conditions, $Y^{\fld}$ equal
endpoint fields, and $Y$ opposite endpoint fields.
All three are strictly positive physical partition functions at both
parities. For \emph{even} $N$, the last expression has the additional
property
\eq
Y_{N,D}(\beta)
 =\langle v_\beta,K_{\beta,D}^{\,N}v_\beta\rangle
 =\sum_\xi |\xi|^N\norm{P_\xi v_\beta}^2,
\lb{Ypositive}
\en
where $P_\xi$ is the spectral projection for a nonzero eigenvalue $\xi$.
This nonnegative weighted moment is the reason for introducing the
opposite-field chain. It will allow us to quantify the overlap of the
boundary vector with the leading spatial eigenvector.

The opposite-field chain has supplied a nonnegative sum at even
lengths, with coefficients fixed by the endpoint vector. We will use
this sum to control the equal-field expectations; the physical
boundary condition in Theorem~\ref{T:boundary} remains unchanged.

\medskip\noindent
{\em Remark:}\/
Our convention is equivalent to that of \cite{OpenAIboundary}, but
puts the ordinary vector on the right in \rlb{Yfieldtransfer}.
That paper uses $v_\beta^{\rm old}(x)=(ih_0)^{r_{\mathrm z}(x)}$
on the same histories, so $v_\beta=\Ux v_\beta^{\rm old}$.
Its opposite-field Hamiltonian has the two field signs reversed.
A global $\pi$ rotation interchanges these signs and preserves the
bulk terms, including $DQ_j$, so the partition functions $Y$ agree.
Since $\Ux$ commutes with every spectral projection of $K_{\beta,D}$,
the boundary-weighted moments and squared overlaps agree as well.
The imported numerical estimates therefore need no change.

\subsection{The symmetry of the leading eigenspace}

\begin{lemma}[The spin rotation on the leading eigenspace]\label{L:character}
Fix $\beta>0$ and $D\in\bbR$. Suppose that $K_{\beta,D}$ has a
simple positive eigenvalue $\lambda_{\beta,D}$ of strictly largest
modulus, with projection $P_{\beta,D}$. If
$c_{\beta,D}:=\norm{P_{\beta,D}v_\beta}^2>0$, then
$\Ux P_{\beta,D}=P_{\beta,D}$.
\end{lemma}
\begin{proof}
Abbreviate $P=P_{\beta,D}$ and $\lambda=\lambda_{\beta,D}$.
Choose a normalized vector $\varphi$ in the range of $P$.
Commutation with $K_{\beta,D}$ implies that $\Ux\varphi$ is in
the same one-dimensional eigenspace. Thus
$\Ux\varphi=\sigma\varphi$; from $\Ux^2=\one$ we get
$\sigma=\pm1$. This scalar action is on the leading eigenspace,
not on the whole vector $v_\beta$ or the whole history space.
The operator-norm convergence of the normalized transfer powers gives
\eq
\lim_{N\uparrow\infty}
 \frac{Y_{N,D}^{\fld}(\beta)}{\lambda^N}
 =\langle\Ux v_\beta,Pv_\beta\rangle
 =\sigma c_{\beta,D}.
\lb{signcharacter}
\en
The left side is a limit of positive physical partition functions.
Since $c_{\beta,D}>0$, it cannot be negative. Hence $\sigma=+1$.
\end{proof}

The uniqueness argument required no spin-rotation symmetry of a
competing state. Here the discrete $\pi$ rotation of the bulk
Hamiltonian is used explicitly to relate the partition functions
with equal and opposite endpoint fields.
Continuous rotation symmetry is not used in this comparison.

\subsection{The additional finite boundary criterion}

In addition to \rlb{STdef}, define
\eq
\tau_D(N,\beta)=\frac{Y_{N,D}(2\beta)}{Y_{N,D}(\beta)^2},
\qquad
\mathcal C_D(k,\beta)=
 \frac{Y_{2k,D}(\beta)}{Z_{k,D}(\beta)Y_{k,D}(\beta)}.
\lb{boundaryratios}
\en
Here $k\ge4$ is even. The first ratio is the physical Gibbs purity
of the opposite-field chain. The second is not used as a purity; it
will certify a nonzero leading overlap using only ordinary traces.
In addition to the periodic criterion, assume
\eq
\tau_D(2k_0,\beta_0)\ge1-\eta,
\qquad \mathcal C_D(k_0,\beta_0)\ge1-\frac\eta2
\qquad(|D|\le D_0).
\lb{boundarycriterion}
\en
For later statements we record the combined criterion as
\eq
\begin{split}
S_D(k_0,\beta_0)\ge1-\eta,\quad T_D(2k_0,\beta_0)&\ge1-\eta,\quad
\tau_D(2k_0,\beta_0)\ge1-\eta,\\
\mathcal C_D(k_0,\beta_0)&\ge1-\frac\eta2.
\end{split}
\lb{finitecriterion}
\en
It involves six physical traces, the three periodic traces and the
analogous three opposite-field traces. The following proposition is
the intermediate target for the boundary analysis.

\begin{proposition}[Finite-chain approximation]\label{P:finiteapproximation}
Fix $h_0>0$ and assume \rlb{finitecriterion}, with the scales and
constants defined in Section~\ref{S:criterion}.
The finite-chain ground states are unique and have gaps greater than
$\gamma_0$ for every even periodic length $N\ge k_0$ and every odd
field-chain length $N\ge2k_0+1$.
At every $\beta_j$, the endpoint-field Gibbs expectations converge,
at fixed temperature, to the periodic state $\nu_{\beta_j,D}$.

For each fixed nonempty finite interval $I=[s,t]$, there is a constant
$C_I>0$, independent of $D,j,N$, such that, for $A\in\calA_I$,
\eq
\begin{split}
\abs{\omega_{N,D}^{b}(A)-\nu_{\beta_j,D}(A)}
 \le\norm A\left\{
 \frac{N}{2k_j}e^{-\gamma_0\beta_j}
 +C_I e^{(3+D_0)m\beta_j}q^{d_N(I)}\right\}.
\end{split}
\lb{finiteapproximation}
\en
Here $m=t-s+1$, $I\subset\Lambda_N$, and
$d_N(I)=\min\{s-\ell_N,r_N-t\}$ is the smaller number of sites
on the two sides of $I$. The bound holds whenever $d_N(I)\ge k_j$,
for even $N$ with $b=\per$ and odd $N$ with $b=\fld$.
All assertions about periodic chains hold under
\rlb{periodiccriterion} alone.
\end{proposition}

The first error in \rlb{finiteapproximation} compares a ground state
with its Gibbs state; the second compares the Gibbs state with
$\nu_{\beta_j,D}$. At lengths of order $\beta_j^2$, both vanish.
Section~\ref{S:boundary} controls the boundary spectral weights.
Section~\ref{S:longchains} proves the local Gibbs and ground-state
estimates and completes this proposition. Section~\ref{S:limit}
then deduces asymptotic translation invariance and invokes the bulk
uniqueness theorem. Appendix~\ref{S:numerics} verifies the criteria.

\section{Concentration of the boundary spectral measure}
\label{S:boundary}

Throughout this section, the combined criterion \rlb{finitecriterion}
is assumed, except where a lemma explicitly states weaker hypotheses.

The comparison between temperatures and the coupled induction in this
section follow \cite[Lemma~4.1 and Propositions~4.2 and~5.3]{OpenAIboundary}.
Our initial overlap test in Lemma~\ref{L:overlapseed} instead uses the
finite partition-function ratio $\mathcal C_D$.

Section~\ref{S:criterion} controls the eigenvalues of the transfer
operator through the periodic partition functions $Z$.
We now use the opposite-field partition functions $Y$ to control the
contribution of each eigenspace to $Y_{N,D}(\beta_j)$.
Even a simple leading eigenvalue would make no contribution if its
eigenvector were orthogonal to the endpoint vector.
Moreover, both the history space and the endpoint vector change with
$\beta$. We will compare their contributions through identities between
physical partition functions, without comparing eigenvectors at two
temperatures. The estimates obtained here will supply the additional
information needed for the equal-field chain.

\subsection{The leading contribution and the relative error}

Recall from Lemma~\ref{L:spectralconcentration} that $P_j$ is the
orthogonal projection onto the one-dimensional eigenspace of
$K_{\beta_j,D}$ corresponding to its leading eigenvalue $\lambda_j$.
Put
\eq
c_j=\norm{P_jv_{\beta_j}}^2,
\lb{boundaryoverlap}
\en
allowing $c_j=0$ until it is proved otherwise.
For even $n\ge2$, self-adjointness and the transfer identity give
\eq
Y_{n,D}(\beta_j)
 =\langle v_{\beta_j},K_{\beta_j,D}^{\,n}v_{\beta_j}\rangle
 =\norm{K_{\beta_j,D}^{\,n/2}v_{\beta_j}}^2>0.
\lb{boundarynorm}
\en
Since $P_j$ commutes with $K_{\beta_j,D}$, the two vectors
$\lambda_j^{n/2}P_jv_{\beta_j}$ and
$K_{\beta_j,D}^{\,n/2}(\one-P_j)v_{\beta_j}$ are orthogonal.
Consequently,
\eq
Y_{n,D}(\beta_j)-\lambda_j^n c_j
 =\norm{K_{\beta_j,D}^{\,n/2}(\one-P_j)v_{\beta_j}}^2\ge0.
\lb{boundaryremainder}
\en
The term $\lambda_j^n c_j$ is the contribution of the leading
eigenspace to $Y_{n,D}(\beta_j)$. We define the relative contribution
from the other eigenspaces by
\eq
e_n(j)=\frac{Y_{n,D}(\beta_j)-\lambda_j^n c_j}{Y_{n,D}(\beta_j)},
\qquad 0\le e_n(j)\le1,
\qquad n\ge2\ \text{even}.
\lb{boundaryquantities}
\en
The inequality $e_n(j)<1$ is equivalent to $c_j>0$.
No lower bound on $c_j$ uniform in $j$ will be needed.
We continue to write the opposite-field Gibbs purity as
$\tau_D(n,\beta_j)$, with its definition in \rlb{boundaryratios}.
These definitions do not yet prove a positive overlap. We first
establish $e_{k_0}(0)<1$, and then control the error as the temperature
changes.

\subsection{The initial overlap from a finite trace ratio}

\begin{lemma}[A finite-trace bound for the overlap]\label{L:overlapseed}
Fix $j\ge0$ and an even $k\ge4$, and suppose
$S_D(k,\beta_j)\ge1-u$, where $0\le u\le1/100$.
Define $\lambda_j$, $P_j$, and $e_k(j)$ as above; the spectral argument
of Lemma~\ref{L:spectralconcentration} applies under this hypothesis
alone. Then
\eq
e_k(j)\le\frac{1-\mathcal C_D(k,\beta_j)}{1-d(u)}.
\lb{Ccertificate}
\en
The assertion does not assume $c_j>0$.
\end{lemma}
\begin{proof}
Let $\xi$ be a nonzero eigenvalue of $K_{\beta_j,D}$ other than
$\lambda_j$, and denote its orthogonal spectral projection by
$P_{\xi,j}$. Set $t=(|\xi|/\lambda_j)^k$.
The proof of Lemma~\ref{L:spectralconcentration}, applied at length
$k$, gives $0\le t\le d(u)$, so
$t-t^2\ge(1-d(u))t$.
Multiply this inequality by
$\lambda_j^k\norm{P_{\xi,j}v_{\beta_j}}^2$ and sum over the
eigenvalues other than $\lambda_j$. The leading contributions cancel,
and we obtain
\eq
Y_{k,D}(\beta_j)-\frac{Y_{2k,D}(\beta_j)}{\lambda_j^k}
 \ge(1-d(u))\{Y_{k,D}(\beta_j)-\lambda_j^k c_j\}.
\lb{Ccontraction}
\en
Here and below, sums with spectral projections run over distinct
nonzero eigenvalues; the projections include any multiplicity.
The kernel of the transfer operator contributes zero to positive powers.

The definition of $\mathcal C_D$ also gives
\eq
\frac{Y_{2k,D}(\beta_j)}{\lambda_j^kY_{k,D}(\beta_j)}
 =\mathcal C_D(k,\beta_j)
       \frac{Z_{k,D}(\beta_j)}{\lambda_j^k}
 \ge\mathcal C_D(k,\beta_j).
\lb{Cratio}
\en
For the last step, use
$Z_{k,D}(\beta_j)=\sum_a|\xi_a|^k\ge\lambda_j^k$.
Divide \rlb{Ccontraction} by the positive $Y_{k,D}(\beta_j)$ and
apply \rlb{Cratio} to get
\eq
(1-d(u))e_k(j)
 \le1-\frac{Y_{2k,D}(\beta_j)}{\lambda_j^kY_{k,D}(\beta_j)}
 \le1-\mathcal C_D(k,\beta_j).
\lb{Cconclusion}
\en
This proves the lemma.
\end{proof}

The last inequality of \rlb{finitecriterion} now implies
\eq
e_{k_0}(0)\le\frac{\eta}{2(1-d(\eta))}<\eta.
\lb{initialoverlap}
\en
In particular, $c_0>0$. Together with
$\tau_D(2k_0,\beta_0)\ge1-\eta$, this supplies the initial boundary
data. The nonzero overlap has so far been proved only at $\beta_0$;
the next argument must establish it anew at each successive
temperature, where the endpoint vector is different.

\subsection{Passing to the next temperature}

First keep the temperature fixed. For even $m>k\ge2$, each
eigenvalue other than $\lambda_j$ satisfies
$|\xi|^m\le(\lambda_j\varrho_j)^{m-k}|\xi|^k$.
Multiplication by its nonnegative boundary weight and summation give
\eq
\begin{split}
0\le Y_{m,D}(\beta_j)-\lambda_j^m c_j
 \le(\lambda_j\varrho_j)^{m-k}
       \{Y_{k,D}(\beta_j)-\lambda_j^k c_j\}.
\end{split}
\lb{weightedcontraction}
\en
The factor $\lambda_j^{m-k}$ is the growth of the leading contribution
through the extra sites. Relative to this growth, the other
contributions contract by $\varrho_j^{m-k}$.
The next lemma combines this observation with a physical
partition-function identity at two temperatures.

\begin{lemma}[Boundary comparison between temperatures]\label{L:boundaryupdate}
Fix $j\ge0$, and suppose that the leading eigenvalues and projections
of $K_{\beta_j,D}$ and $K_{\beta_{j+1},D}$ are as in
Lemma~\ref{L:spectralconcentration} and that $r_j\le1$.
Let $m>k\ge2$ be even. Suppose $e_k(j)<1$ and $0<q'<1$ with
$\varrho_{j+1}^{m-k}\le q'$. Then
\eq
e_m(j+1)\le F\bigl(q',e_k(j),\tau_D(m,\beta_j)\bigr),
\qquad
F(q',e,t)=\frac{q'}{1-q'}
 \left\{\frac1{(1-e)^2t}-1\right\}.
\lb{Fupdate}
\en
A bound less than one on the right proves $c_{j+1}>0$.
\end{lemma}
\begin{proof}
At the new temperature, \rlb{weightedcontraction} implies
\eq
\begin{split}
(1-q')\{Y_{m,D}(\beta_{j+1})-\lambda_{j+1}^m c_{j+1}\}
 \le q'\{\lambda_{j+1}^{m-k}Y_{k,D}(\beta_{j+1})
                         -Y_{m,D}(\beta_{j+1})\}.
\end{split}
\lb{newtemperaturecontraction}
\en
To see this rearrangement, subtract $q'$ times the left-hand
remainder in \rlb{weightedcontraction} from both sides of that
inequality, with $j$ replaced by $j+1$.
Dividing by $(1-q')Y_{m,D}(\beta_{j+1})>0$ gives
\eq
e_m(j+1)\le\frac{q'}{1-q'}
 \left\{\frac{\lambda_{j+1}^{m-k}Y_{k,D}(\beta_{j+1})}
                    {Y_{m,D}(\beta_{j+1})}-1\right\}.
\lb{Ffirst}
\en

The definition of the physical Gibbs purity gives the exact identity
\eq
Y_{n,D}(\beta_{j+1})
 =\tau_D(n,\beta_j)Y_{n,D}(\beta_j)^2.
\lb{Ytemperature}
\en
Apply it at $n=k$ and $n=m$. Keeping the leading-eigenvalue factors
explicit, we obtain
\eq
\frac{\lambda_{j+1}^{m-k}Y_{k,D}(\beta_{j+1})}
     {Y_{m,D}(\beta_{j+1})}
 =r_j^{m-k}\frac{\tau_D(k,\beta_j)}{\tau_D(m,\beta_j)}
  \left\{\frac{\lambda_j^{m-k}Y_{k,D}(\beta_j)}
                     {Y_{m,D}(\beta_j)}\right\}^2.
\lb{Ytemperatureratio}
\en
The old overlap is positive because $e_k(j)<1$. Therefore
\eq
\frac{\lambda_j^{m-k}Y_{k,D}(\beta_j)}{Y_{m,D}(\beta_j)}
 \le\frac{Y_{k,D}(\beta_j)}{\lambda_j^k c_j}
 =\frac1{1-e_k(j)}.
\lb{oldoverlapratio}
\en
Use this bound, $r_j\le1$, and $\tau_D(k,\beta_j)\le1$ in
\rlb{Ytemperatureratio}, and then substitute into \rlb{Ffirst}.
This proves \rlb{Fupdate}. Only the old overlap $c_j$ has been
divided by; $c_{j+1}$ may still be zero during the argument.
\end{proof}

The lemma bounds the new overlap error using quantities at the old
temperature and the already controlled transfer eigenvalues. A bound
below one proves the new overlap is nonzero. To continue indefinitely,
we must also control the opposite-field Gibbs purity appearing on the
right; we do this next.

\subsection{The coupled boundary estimates}

We next write the two coupled updates explicitly. Suppose
$e_{k_j}(j)\le E_j<1$ and
$\tau_D(2k_j,\beta_j)\ge1-\zeta_j$, with $0\le\zeta_j<1$.
Since $k_{j+1}=2k_j$, \rlb{residualmoment} implies
$\varrho_{j+1}^{k_j}\le\sqrt{d(u_{j+1})}$.
Lemma~\ref{L:boundaryupdate}, with $k=k_j$ and $m=2k_j$, gives
\eq
e_{k_{j+1}}(j+1)
 \le F\bigl(\sqrt{d(u_{j+1})},E_j,1-\zeta_j\bigr).
\lb{boundaryE}
\en

For the physical purity, cancellation of partition functions gives
\eq
\frac{\tau_D(4k_j,\beta_j)}{\tau_D(2k_j,\beta_j)}
 =r_j^{2k_j}
\frac{Y_{4k_j,D}(\beta_{j+1})}
        {\lambda_{j+1}^{2k_j}Y_{2k_j,D}(\beta_{j+1})}
 \left\{\frac{\lambda_j^{2k_j}Y_{2k_j,D}(\beta_j)}
                    {Y_{4k_j,D}(\beta_j)}\right\}^2.
\lb{tauidentity}
\en
Each factor has a separate bound. At the old temperature,
$Y_{4k_j,D}(\beta_j)\le\lambda_j^{2k_j}Y_{2k_j,D}(\beta_j)$,
by comparison of the nonnegative even powers of each eigenvalue.
Thus the squared ratio is at least one.
At the new temperature, retaining the leading contribution in
$Y_{4k_j,D}(\beta_{j+1})$ gives
\eq
\frac{Y_{4k_j,D}(\beta_{j+1})}
     {\lambda_{j+1}^{2k_j}Y_{2k_j,D}(\beta_{j+1})}
 \ge\frac{\lambda_{j+1}^{2k_j}c_{j+1}}
           {Y_{2k_j,D}(\beta_{j+1})}
 =1-e_{k_{j+1}}(j+1).
\lb{newpurityratio}
\en
If \rlb{boundaryE} has given $e_{k_{j+1}}(j+1)\le E_{j+1}<1$,
we can use \rlb{rbound} in \rlb{tauidentity} to conclude
\eq
\tau_D(4k_j,\beta_j)
 \ge(1-\zeta_j)(1-u_j)(1-u_{j+1})(1-E_{j+1}).
\lb{taubefore}
\en
At the fixed length $4k_j$, doubling the inverse temperature squares
and normalizes the Gibbs probabilities. Apply \rlb{squaringpurity},
then use \rlb{taubefore} and the monotonicity of $h$ to obtain
\eqa
\tau_D(2k_{j+1},\beta_{j+1})
 &=\tau_D(4k_j,2\beta_j)\nl
 &\ge h\bigl(\tau_D(4k_j,\beta_j)\bigr)\nl
 &\ge h\bigl((1-\zeta_j)(1-u_j)(1-u_{j+1})(1-E_{j+1})\bigr).
\lb{boundaryzeta}
\ena
The equality uses $k_{j+1}=2k_j$ and $\beta_{j+1}=2\beta_j$.
Equations~\rlb{boundaryE} and \rlb{boundaryzeta} are the two coupled
updates: the new overlap estimate is inserted into the lower bound
for the new Gibbs purity.

Recall from \rlb{errormatching} that $u_j\le x_j$.
The following proposition proves the two additional bounds for the
boundary quantities.

\begin{proposition}[The boundary error bounds]\label{P:quadratic}
Under \rlb{finitecriterion}, for every $j\ge0$,
\eq
e_{k_j}(j)\le x_j,\qquad
\tau_D(2k_j,\beta_j)\ge1-x_j.
\lb{allerrors}
\en
In particular, $c_j>0$ and $\Ux P_j=P_j$ at every temperature $\beta_j$.
\end{proposition}
\begin{proof}
The case $j=0$ follows from \rlb{finitecriterion} and
\rlb{initialoverlap}. Suppose both bounds hold at level $j$, and write
$x=x_j\le\eta\le1/100$, as follows from \rlb{errorscales}.
Using $u_j\le x_j$ and $u_{j+1}=6u_j^2$, we have
\eq
\sqrt{d(u_{j+1})}
 \le\left(\frac{6x^2}{2-18x^2}\right)^{1/2}\le2x.
\lb{qstep}
\en
The function $F$ increases in its first two variables and decreases in
its third. Moreover,
\eqa
\frac{F(2x,x,1-x)}{x^2}
 &=\frac{2}{1-2x}
   \{(1-x)^{-1}+(1-x)^{-2}+(1-x)^{-3}\}\nl
 &\le\frac{6}{(1-2x)(1-x)^3}<10.
\lb{Fscalar}
\ena
Thus \rlb{boundaryE} gives
$e_{k_{j+1}}(j+1)\le10x^2=x_{j+1}$ and proves $c_{j+1}>0$,
where the equality follows directly from \rlb{errorscales}.
For the second update, take $\zeta_j=x$ and $E_{j+1}=10x^2$
in \rlb{boundaryzeta}. Each factor of the product in
\rlb{taubefore} lies in $[0,1]$. The defect of the product is
therefore at most the sum of their
defects, namely $x+x+6x^2+10x^2=2x+16x^2$.
Using \rlb{boundaryzeta} and the monotonicity of $h$, we obtain
\eqa
\tau_D(2k_{j+1},\beta_{j+1})
 &\ge h(1-2x-16x^2)\nl
 &=1-\frac{(2x+16x^2)^2}{2(1-2x-16x^2)^2}\nl
 &\ge1-3x^2\ge1-10x^2=1-x_{j+1}.
\lb{boundaryhstep}
\ena
The final equality follows from \rlb{errorscales}.
For the bound by $1-3x^2$, use the scalar estimate
\eq
\frac{(2+16x)^2}{2(1-2x-16x^2)^2}<3,
\qquad 0\le x\le1/100.
\lb{tauscalar}
\en
Its numerator increases and its positive denominator decreases on this
interval, and substitution of $x=1/100$ proves the bound.
This proves the second bound in \rlb{allerrors} at level $j+1$,
and completes the induction.
Lemma~\ref{L:character} now gives $\Ux P_j=P_j$ for every $j$.
\end{proof}

The endpoint vector now has a nonzero leading overlap at every
$\beta_j$. Its relative error and the opposite-field Gibbs-purity
defect are both at most $x_j$. These conclusions do not require a
positive lower bound on $c_j$ uniform in $j$ or $D$.

\subsection{The contribution of the other eigenspaces}

Under \rlb{finitecriterion}, the periodic estimate
\rlb{operatorconcentration} and Proposition~\ref{P:quadratic} give
one further estimate needed for local expectations.
The weighted estimate is
\eq
\sum_{\xi\ne\lambda_j}|\xi|^s
          \norm{P_{\xi,j}v_{\beta_j}}^2
 \le\frac{10}{99}\,c_j\lambda_j^s q^s,
\qquad s\ge k_j.
\lb{weightedall}
\en
To prove it, bound $|\xi|^s$ by
$(\lambda_j q)^{s-k_j}|\xi|^{k_j}$ and use
\eq
\frac{Y_{k_j,D}(\beta_j)-\lambda_j^{k_j}c_j}
     {\lambda_j^{k_j}c_j}
 =\frac{e_{k_j}(j)}{1-e_{k_j}(j)}
 \le\frac{x_j}{1-x_j}.
\lb{weightedrelative}
\en
Substitution of $x_j=q^{k_j}/10$ and $x_j\le\eta\le1/100$, again
from the definitions, gives \rlb{weightedall}.
The scalar absolute values in that sum allow $s$ to be odd or even;
for even $s$ its left side is
$Y_{s,D}(\beta_j)-\lambda_j^s c_j$.
The estimate is relative to the leading weight $c_j$, not to the norm
of the whole endpoint vector. It therefore remains useful even when
$c_j$ is small.

The periodic trace estimate \rlb{residualall} and the weighted
estimate \rlb{weightedall} now give the two forms of spatial control
needed for local expectations. Both use the same $q<1$, independent
of $j$ and $D$. The next section applies them to the exact insertion
formulas and then compares the Gibbs and ground-state expectations.

\section{Local expectations and low-temperature boundary estimates}
\label{S:longchains}\label{S:insertions}

Throughout this section assume \rlb{finitecriterion}. The periodic
Gibbs estimate \rlb{exactpererror} is already available. We derive its
endpoint counterpart, extend physical purity to longer chains, and
complete Proposition~\ref{P:finiteapproximation}.

\subsection{The endpoint fields disappear in the bulk}

For $I=[s,t]$, use the distances $d_{\rm L},d_{\rm R}$ in
\rlb{distances}, so that $N=d_{\rm L}+m+d_{\rm R}$.
Move the left transfer power in \rlb{fieldGibbsraw} to the first
slot using self-adjointness. The exact expectation is
\eq
\rho_{N,\beta_j,D}^{\fld}(A)
 =\frac{\langle K_{\beta_j,D}^{\,d_{\rm L}}\Ux v_{\beta_j},
       \mathcal B_{A,\beta_j,D}
       K_{\beta_j,D}^{\,d_{\rm R}}v_{\beta_j}\rangle}
        {Y_{N,D}^{\fld}(\beta_j)}.
\lb{fieldinsertion}
\en
We compare each propagated endpoint vector with its leading component.
For $2d\ge k_j$, \rlb{weightedall} gives
\eq
\frac{\norm{K_{\beta_j,D}^{\,d}v_{\beta_j}
                  -\lambda_j^dP_jv_{\beta_j}}}
     {\lambda_j^d\sqrt{c_j}}
 \le\sqrt{\frac{10}{99}}\,q^d.
\lb{filteredvector}
\en
To prove this, square the numerator and use the orthogonal spectral
decomposition. The result is
$\sum_{\xi\ne\lambda_j}|\xi|^{2d}
\norm{P_{\xi,j}v_{\beta_j}}^2$, to which \rlb{weightedall} applies.
The same bound holds with $\Ux v_{\beta_j}$ in place of
$v_{\beta_j}$, since $\Ux$ is unitary, commutes with the transfer
operator, and fixes the leading eigenspace.

The denominator satisfies, for every integer $N\ge k_j$,
\eq
\abs{Y_{N,D}^{\fld}(\beta_j)-\lambda_j^Nc_j}
 \le\frac{10}{99}\,c_j\lambda_j^Nq^N.
\lb{fielddenominator}
\en
To obtain this estimate, expand the inner product in
\rlb{Yfieldtransfer} using the spectral projections of the transfer operator.
The leading term is $\lambda_j^Nc_j$. For each other eigenvalue,
Cauchy--Schwarz and commutation with $\Ux$ give
$|\langle\Ux P_{\xi,j}v_{\beta_j},P_{\xi,j}v_{\beta_j}\rangle|
\le\norm{P_{\xi,j}v_{\beta_j}}^2$.
Now use \rlb{weightedall} with $s=N$.

Since $P_j$ has rank one,
$\langle P_jv_{\beta_j},\mathcal B_{A,\beta_j,D}P_jv_{\beta_j}\rangle
=c_j\Tr(P_j\mathcal B_{A,\beta_j,D})$.
The leading term in the numerator of \rlb{fieldinsertion} is therefore
\eq
\lambda_j^{d_{\rm L}+d_{\rm R}}c_j
       \Tr(P_j\mathcal B_{A,\beta_j,D})
 =\lambda_j^Nc_j\nu_{\beta_j,D}(A).
\lb{fieldleadingnumerator}
\en
Expand both propagated vectors as this leading part plus their
remainders. The other three terms contain one or two remainders.
Applying \rlb{filteredvector} and Cauchy--Schwarz bounds their sum,
in absolute value, by
\eq
\begin{split}
 c_j\lambda_j^{N-m}\norm{\mathcal B_{A,\beta_j,D}}
 \left\{\sqrt{\frac{10}{99}}(q^{d_{\rm L}}+q^{d_{\rm R}})
               +\frac{10}{99}q^{d_{\rm L}+d_{\rm R}}\right\}.
\end{split}
\lb{fieldnumeratorerror}
\en
Combining \rlb{fielddenominator}--\rlb{fieldnumeratorerror} with
\rlb{Bnormalizednorm}, we obtain, for $2d_{\rm L},2d_{\rm R}\ge k_j$,
\eq
\begin{split}
&\abs{\rho_{N,\beta_j,D}^{\fld}(A)-\nu_{\beta_j,D}(A)}\\
&\quad\le
 \frac{\norm A}{1-(10/99)q^N}
 \left[
 e^{(3+D_0)m\beta_j}
 \left\{\sqrt{\frac{10}{99}}(q^{d_{\rm L}}+q^{d_{\rm R}})
             +\frac{10}{99}q^{d_{\rm L}+d_{\rm R}}\right\}
       +\frac{10}{99}q^N
 \right].
\end{split}
\lb{exactfielderror}
\en
All factors $\lambda_j^Nc_j$ have canceled.
This is why we needed an estimate relative to the leading boundary
contribution, rather than an absolute bound on the endpoint-vector norm.

Recall $d_N(I)=\min(d_{\rm L},d_{\rm R})$ from
Proposition~\ref{P:finiteapproximation}.
The periodic and field estimates imply, with a constant $C_I$
independent of $j,D,N$,
\eq
\abs{\rho_{N,\beta_j,D}^{b}(A)-\nu_{\beta_j,D}(A)}
 \le C_I\norm A\,e^{(3+D_0)m\beta_j}q^{d_N(I)},
\qquad d_N(I)\ge k_j,
\lb{thermalcomparison}
\en
for even periodic lengths and all field-chain lengths.
In particular, for even $N_{\per}$ and odd $N_{\fld}$ satisfying these
distance conditions,
\eq
\begin{split}
&\abs{\rho_{N_{\per},\beta_j,D}^{\per}(A)
          -\rho_{N_{\fld},\beta_j,D}^{\fld}(A)}\\
&\qquad\le C_I\norm A\,e^{(3+D_0)m\beta_j}
       \{q^{d_{N_{\per}}(I)}+q^{d_{N_{\fld}}(I)}\}.
\end{split}
\lb{directcomparison}
\en
This is the finite-size, finite-inverse-temperature comparison.
The two lengths may, for example, be $2L$ and $2L+1$.
Both Gibbs expectations are now close to the same state
$\nu_{\beta_j,D}$, with a quantitative error. To reach the
finite-chain approximation in Proposition~\ref{P:finiteapproximation},
we still have to compare these Gibbs states with the ground states at
the same lengths and temperatures.

\subsection{Gibbs purity at longer chain lengths}

The iteration initially controls lengths proportional to $\beta_j$.
We first extend physical purity to much longer chains, allowing
$N$ of order $\beta_j^2$ in the final comparison.

For the physical-purity argument we also need the Gibbs density
matrix itself. We denote it by
\eq
\hat\rho_{N,\beta,D}^{b}
 =\frac{e^{-\beta H_N^b(D)}}{\Tr_{\mathcal H_N}e^{-\beta H_N^b(D)}},
\qquad
\rho_{N,\beta,D}^{b}(A)
 =\Tr_{\mathcal H_N}(\hat\rho_{N,\beta,D}^{b}A).
\lb{Gibbsdensitymatrix}
\en
Thus $\rho_{N,\beta,D}^{b}(\cdot)$ always denotes an expectation
functional, whereas $\hat\rho_{N,\beta,D}^{b}$ is an operator.

\begin{proposition}\label{P:longpurity}
Under \rlb{finitecriterion}, for every $j\ge0$ and every $N\ge2k_j$,
\eq
\Tr\bigl[(\hat\rho_{N,\beta_j,D}^{b})^2\bigr]
 \ge1-5\frac{N}{2k_j}x_j,
\lb{longpurity}
\en
where $b=\per$ at even $N$ and $b=\fld$ at every integer $N$.
There is no upper restriction on $N$; the estimate is useful when its
right-hand side is close to one. For $b=\per$, only
\rlb{periodiccriterion} is required.
\end{proposition}
\begin{proof}
For even $N\ge2k_j$, the nonnegative even powers and
\rlb{residualall} imply
\eq
\lambda_j^N\le Z_{N,D}(\beta_j)\le\lambda_j^N(1+x_j),
\qquad
Z_{N,D}(\beta_{j+1})\ge\lambda_{j+1}^N.
\lb{longpartitionbounds}
\en
To see the upper bound, use $q^N/10\le q^{k_j}/10=x_j$,
where the final equality follows directly from \rlb{gamma0},
\rlb{scales}, and \rlb{errorscales}.
The definition of $T_D$ and \rlb{rbound} therefore give
\eq
T_D(N,\beta_j)
 \ge\frac{r_j^N}{(1+x_j)^2}
 \ge\frac{\{(1-u_j)(1-u_{j+1})\}^{N/(2k_j)}}{(1+x_j)^2}.
\lb{longperiodic}
\en
Bernoulli's inequality, $(1+x)^{-2}\ge1-2x$, and
$u_{j+1}\le u_j\le x_j$ yield
\eq
1-T_D(N,\beta_j)
 \le2x_j+\frac{N}{2k_j}(u_j+u_{j+1})
 \le4\frac{N}{2k_j}x_j.
\lb{longperiodicerror}
\en

For the equal-field chain we use the comparison of partition functions in
\cite[Lemma~6.1 and the proof of Proposition~6.2]{OpenAIboundary},
retaining its length dependence beyond the doubling interval.
We first compare with an even reference length $n$. For any integer $N\ge n\ge2$ with $n$ even,
\eq
\abs{Y_{N,D}^{\fld}(\beta_j)-\lambda_j^Nc_j}
 \le\lambda_j^{N-n}
        \{Y_{n,D}(\beta_j)-\lambda_j^n c_j\}.
\lb{signedcomparison}
\en
To prove this inequality, expand in the eigenspaces of the transfer
operator, as for \rlb{fielddenominator}. An eigenvalue $\xi\ne\lambda_j$
contributes at most $|\xi|^N\norm{P_{\xi,j}v_{\beta_j}}^2$ in
absolute value. Use $|\xi|^N\le\lambda_j^{N-n}|\xi|^n$ and sum.
The evenness of $n$ identifies the resulting sum with the remainder
on the right. This argument includes negative eigenvalues and odd $N$.

Now use the same reference length $n=k_{j+1}=2k_j$ at both
temperatures. At the old temperature, \rlb{signedcomparison} gives
\eq
Y_{N,D}^{\fld}(\beta_j)
 \le\lambda_j^{N-k_{j+1}}Y_{k_{j+1},D}(\beta_j).
\lb{oldfieldpartition}
\en
At the new temperature it gives
\eq
\begin{split}
Y_{N,D}^{\fld}(\beta_{j+1})
 &\ge\lambda_{j+1}^{N-k_{j+1}}
       \{2\lambda_{j+1}^{k_{j+1}}c_{j+1}
                              -Y_{k_{j+1},D}(\beta_{j+1})\}\\
 &\ge(1-2x_{j+1})\lambda_{j+1}^{N-k_{j+1}}
                            Y_{k_{j+1},D}(\beta_{j+1}).
\end{split}
\lb{newfieldpartition}
\en
The last step uses $e_{k_{j+1}}(j+1)\le x_{j+1}$.
The old partition function and its upper bound are both positive,
so we may square that upper bound in the denominator of the physical
Gibbs purity. For every integer $N\ge2k_j$, we conclude
\eqa
\Tr\bigl[(\hat\rho_{N,\beta_j,D}^{\fld})^2\bigr]
 &=\frac{Y_{N,D}^{\fld}(\beta_{j+1})}
         {Y_{N,D}^{\fld}(\beta_j)^2}\nl
 &\ge(1-2x_{j+1})\tau_D(2k_j,\beta_j)r_j^{N-2k_j}\nl
 &\ge(1-2x_{j+1})(1-x_j)
       \{(1-u_j)(1-u_{j+1})\}^{N/(2k_j)}.
\lb{longfield}
\ena
For the final inequality, use $r_j\le1$ to obtain
$r_j^{N-2k_j}\ge r_j^N$, and then apply the lower bound in
\rlb{rbound}.
Bernoulli's inequality gives
\eq
1-\Tr\bigl[(\hat\rho_{N,\beta_j,D}^{\fld})^2\bigr]
 \le2x_{j+1}+x_j+\frac{N}{2k_j}(u_j+u_{j+1})
 \le5\frac{N}{2k_j}x_j.
\lb{longfielderror}
\en
Together with \rlb{longperiodicerror}, this proves the proposition.
\end{proof}

The physical-purity estimate now reaches every longer chain, with a
loss proportional only to $N/(2k_j)$. In particular, at lengths of
order $k_j^2$ its defect is of order $k_jx_j$ and still tends to zero.
This is the length range needed for the final comparison; it is not
obtained from a finite-chain gap alone.

\subsection{Finite-chain gaps}

Write $\gamma_N^b(D)$ for the difference between the two lowest
eigenvalues of $H_N^b(D)$, counted with multiplicity.
The implication \rlb{puritygapbasic} converts a physical-purity
bound into simplicity and a gap. The periodic assertion was already
proved in Lemma~\ref{L:periodicgap}, under the periodic criterion
alone. We prove the equal-field assertion here.

For the equal-field chain and $2k_j\le N\le4k_j$, retain the first
lower bound in \rlb{longfield} and use
$r_j^{N-2k_j}\ge r_j^{2k_j}\ge(1-u_j)(1-u_{j+1})$.
The resulting purity defect is at most
\eq
2x_{j+1}+x_j+u_j+u_{j+1}
 \le2x_j+26x_j^2\le4x_j.
\lb{fieldinterval}
\en
It follows that the ground state is unique and
\eq
e^{-\beta_j\gamma_N^{\fld}(D)}
 \le\frac{4x_j}{1-4x_j}<(10\eta)^{2^j}
 =e^{-\gamma_0\beta_j}.
\lb{fieldgap}
\en
For the strict inequality, use $x_j\le1/100$ and
$10x_j=(10\eta)^{2^j}$; the final equality in \rlb{fieldgap}
follows directly from \rlb{gamma0} and \rlb{scales}.
The intervals $[2k_j,4k_j]$ cover every integer $N\ge2k_0$,
including every odd $N\ge2k_0+1$. This proves the finite-gap
assertions of Proposition~\ref{P:finiteapproximation}.
Both selected families therefore have unique ground states and a
common gap bound, uniformly in $D$. We next use the stronger purity
estimates to compare these ground states with the Gibbs states.

\subsection{From Gibbs states to ground states}

If $p_0$ is the largest eigenvalue of a Gibbs density matrix
$\hat\rho$, then $\Tr\hat\rho^2\le p_0$.
For a normalized unique ground vector $\psi_0$,
\eq
\norm{\hat\rho-|\psi_0\rangle\langle\psi_0|}_1
 =2(1-p_0)\le2(1-\Tr\hat\rho^2).
\lb{puritytostate}
\en
Applying Proposition~\ref{P:longpurity} gives
\eq
\abs{\omega_{N,D}^{b}(A)-\rho_{N,\beta_j,D}^{b}(A)}
 \le10\norm A\,\frac{N}{2k_j}x_j,
\lb{groundthermal}
\en
for the allowed lengths $N\ge2k_j$.
In particular, when $N$ is of order $\beta_j^2$, the right side is of
order $\beta_j e^{-\gamma_0\beta_j}$.
The purity estimate supplies this bound without a factor equal to the
full Hilbert-space dimension $3^N$.

\medskip\noindent
{\em Completion of the proof of Proposition~\ref{P:finiteapproximation}.}
The state $\nu_{\beta_j,D}$ was constructed from periodic chains
in Lemma~\ref{L:fixedtemperature}. The bound \rlb{thermalcomparison}
identifies it also as the endpoint-field Gibbs limit at fixed
$\beta_j$. The finite-chain gaps have now been proved.
For the remaining estimate, use
\eq
\begin{split}
\abs{\omega_{N,D}^{b}(A)-\nu_{\beta_j,D}(A)}
 &\le\abs{\omega_{N,D}^{b}(A)-\rho_{N,\beta_j,D}^{b}(A)}\\
 &\quad+\abs{\rho_{N,\beta_j,D}^{b}(A)-\nu_{\beta_j,D}(A)}.
\end{split}
\lb{approximationtriangle}
\en
Substituting \rlb{groundthermal} and \rlb{thermalcomparison}, and
using $10x_j=e^{-\gamma_0\beta_j}$, which follows directly from
\rlb{errorscales}, \rlb{gamma0}, and \rlb{scales}, proves
\rlb{finiteapproximation}. The condition $d_N(I)\ge k_j$ implies
$N\ge2k_j+m\ge2k_j$, so both estimates apply. This completes the
proof of the proposition.\hfill\qedm

We have proved the approximation stated in
Proposition~\ref{P:finiteapproximation}: the ground-state
expectations for both boundary conditions are close to
$\nu_{\beta_j,D}$, with two explicit errors. The state
$\nu_{\beta_j,D}$ is exactly translation invariant by
Lemma~\ref{L:fixedtemperature}. The next section uses this
invariance, together with the approximation, to identify the
limits of ground states of chains with endpoint fields.

\section{Limits of ground states with endpoint fields}
\label{S:limit}

For the chain with endpoint fields, assume \rlb{finitecriterion}.
Theorem~\ref{T:main} has already been proved using only the periodic
criterion. We now consider limits of subsequences of the expectations
$\omega_{2L+1,D}^{\fld}$. We show that every such limit is
translation invariant. Since it is also a ground state of the
infinite chain, Theorem~\ref{T:main} then identifies it with
$\omega_D$.

\subsection{Asymptotic translation invariance}

Let $A$ be local, and choose an interval $I$ containing the supports
of both $A$ and $\tau_1(A)$. Put $m=|I|$ and
$d_N(I)=\min\{s-\ell_N,r_N-t\}$ for $I=[s,t]$.
The periodic Gibbs limit is exactly translation invariant:
\eq
\nu_{\beta_j,D}(\tau_1(A))=\nu_{\beta_j,D}(A).
\lb{nutranslation}
\en
Apply Proposition~\ref{P:finiteapproximation} to the two observables,
subtract, and use \rlb{nutranslation}. For sufficiently large odd
$N$ with $d_N(I)\ge k_j$,
\eq
\begin{split}
&|\omega_{N,D}^{\fld}(\tau_1(A))-\omega_{N,D}^{\fld}(A)|\\
&\quad\le2\norm A\left\{
\frac{N}{2k_j}e^{-\gamma_0\beta_j}
+C_Ie^{(3+D_0)m\beta_j}q^{d_N(I)}\right\}.
\end{split}
\lb{asymptoticTIestimate}
\en
The two occurrences of $\nu_{\beta_j,D}(A)$ cancel by
\rlb{nutranslation}. Thus this step uses the translation invariance
of $\nu_{\beta_j,D}$, not its convergence as $j\uparrow\infty$.

For every sufficiently large $N$, choose the unique $j=j(N)$ with
\eq
k_j^2\le N<k_{j+1}^2=4k_j^2.
\lb{quadraticlengths}
\en
For a fixed central support, $d_N(I)\ge N/3$ for large $N$.
Also $N/(2k_j)\le2k_j$, and
$\beta_j=(\beta_0/k_0)k_j$. Thus both terms on the right of
\rlb{asymptoticTIestimate} tend to zero, uniformly in $D$.
In particular,
\eq
\lim_{\substack{N\uparrow\infty\\N\ {\rm odd}}}
 \sup_{|D|\le D_0}
 |\omega_{N,D}^{\fld}(\tau_1(A))-\omega_{N,D}^{\fld}(A)|=0.
\lb{asymptoticTI}
\en
It follows that every local limit of a subsequence of
$\omega_{2L+1,D}^{\fld}$ is translation invariant. It is also a
ground state: for fixed local $A$, the finite-chain
ground-state inequality involves $[H(D),A]$ once the endpoints have
receded beyond its support and adjacent sites. Theorem~\ref{T:main}
identifies every such limit with $\omega_D$. Compactness then gives
the limit \rlb{fieldlimit} along the full sequence of odd lengths.
The periodic limit \rlb{PBClimit} was proved in
Section~\ref{S:bulkconsequences}.

The argument identifies every subsequential limit, not just a
limit along one specially chosen subsequence. This is why the
whole sequence converges.

\subsection{The quantitative rate and uniformity in the parameter}
\label{S:quantitativerate}

We prove the bounds \rlb{localratePBC} and \rlb{localrateField}
together. For the periodic case, only \rlb{periodiccriterion} is
needed: the periodic part of Proposition~\ref{P:finiteapproximation}
follows from \rlb{exactpererror}, \rlb{longperiodicerror},
and the Gibbs-to-ground-state bound \rlb{puritytostate}.
None of these uses the boundary-field criterion.

Combine Proposition~\ref{P:finiteapproximation} and \rlb{nurate}.
For $j=j(N)$ in \rlb{quadraticlengths}, $b\in\{\per,\fld\}$, and the corresponding lengths (even for
$b=\per$ and odd for $b=\fld$), we obtain
\eq
\begin{split}
|\omega_{N,D}^{b}(A)-\omega_D(A)|
 \le C_I\norm A\bigl\{&k_j e^{-\gamma_0\beta_j}
       +\sqrt{k_ju_j}\\
       &+e^{(3+D_0)m\beta_j}q^{N/3}
       +e^{(3+D_0)m\beta_j/2}q^{k_j^2/2}\bigr\}.
\end{split}
\lb{combinedlocalrate}
\en
The support and length conditions hold for all sufficiently large $N$.
The first two terms decay exponentially in $k_j$, apart from the
displayed powers. The last two have a negative exponent quadratic
in $k_j$ and a positive exponent only linear in $k_j$.
Since $\sqrt N/2<k_j\le\sqrt N$, the right side is bounded by
$C_I\norm A\sqrt N\,e^{-c\sqrt N}$ after increasing $C_I$ and
choosing a sufficiently small $c>0$, independent of $I$.
For $N=2L$ or $N=2L+1$ and $L\ge1$, the inequalities
$\sqrt L\le\sqrt N\le\sqrt3\sqrt L$ then give
\rlb{localratePBC} and \rlb{localrateField}, after a further
increase of $C_I$. The length threshold may depend on $I$.
Every constant is uniform for $|D|\le D_0$ and on the unit ball
of $\calA_I$.

This completes Theorem~\ref{T:periodic} under the periodic criterion
and Theorem~\ref{T:boundary} under the combined criterion.
Both criteria are verified in the appendix. The result includes restoration
of the bulk symmetries broken by the equal endpoint fields. The limiting states are identified by the uniqueness theorem;
the convergence rate follows from its quantitative estimate
\rlb{nurate}.
\section{Open chains without endpoint magnetic fields}
\label{S:free}

In this section assume only the periodic criterion
\rlb{periodiccriterion}. We prove Theorem~\ref{T:free} without
assuming a four-dimensional low-energy space or any finite-chain
free-end gap. The distinction from the endpoint-field proof is that
physical Gibbs purity need not approach one. A strictly positive
lower bound will suffice because the limiting bulk state is pure.

\subsection{The empty-history vector and its leading overlap}

For the free-end Hamiltonian \rlb{Hfree}, write
\eq
W_{N,D}(\beta)=\Tr e^{-\beta H_N^{\free}(D)},\qquad
v_\beta^{\free}(x)=\mathbf1_{\{x=\varnothing\}},\qquad\norm{v_\beta^{\free}}=1.
\lb{freepartition}
\en
The boundary expansion with no field events gives
\eq
W_{N,D}(\beta)=\langle v_\beta^{\free},K_{\beta,D}^Nv_\beta^{\free}\rangle.
\lb{freetransfer}
\en
Define $c_j^{\free}=\norm{P_jv_{\beta_j}^{\free}}^2$. This overlap has an
immediate lower bound. The closing bond is bounded below by
$-2\one$, so $H_N^{\per}(D)\ge H_N^{\free}(D)-2\one$.
Eigenvalue comparison gives
$W_{N,D}(\beta_j)\ge e^{-2\beta_j}Z_{N,D}(\beta_j)$.
Divide by $\lambda_j^N$ and let even $N\uparrow\infty$.
The spectral estimate \rlb{residualall} gives
$Z_{N,D}(\beta_j)/\lambda_j^N\to1$, while
\rlb{freetransfer} and \rlb{operatorconcentration} give
$W_{N,D}(\beta_j)/\lambda_j^N\to c_j^{\free}$. Therefore
\eq
c_j^{\free}\ge e^{-2\beta_j}>0.
\lb{freeoverlap}
\en
In particular, \rlb{freetransfer} and \rlb{operatorconcentration}
yield, at either parity,
\eq
\left|\frac{W_{N,D}(\beta_j)}{\lambda_j^Nc_j^{\free}}-1\right|
 \le e^{2\beta_j}q^N.
\lb{freerelative}
\en
For even $N$, the expression inside the absolute value is
nonnegative, since all the even spectral contributions are
nonnegative. This exponentially crude overlap bound is sufficient
at the long lengths used below.

\subsection{A modified length--temperature iteration}

Define the physical free-end purity by
\eq
\tau_D^{\free}(N,\beta)=\frac{W_{N,D}(2\beta)}{W_{N,D}(\beta)^2}.
\lb{freepurity}
\en
For any finite probability list $(p_a)$, H\"older's inequality gives
$(\sum_a p_a^2)^3\le(\sum_a p_a)^2\sum_a p_a^4$.
Since $\sum_a p_a=1$, normalizing the squares cannot decrease purity.
Consequently,
\eq
\tau_D^{\free}(N,2\beta)\ge\tau_D^{\free}(N,\beta).
\lb{puritymonotone}
\en
Unlike the quadratic purification estimate near one, this inequality
is useful at any positive purity.

Choose a fixed even integer $R$ so large that
$Rk_0\log(1/q)>4\beta_0$, and put
\eq
n_j=Rk_j,\qquad
\delta_j=e^{4\beta_j}q^{n_j}
 =\exp\{-\bigl(Rk_0\log(1/q)-4\beta_0\bigr)2^j\}.
\lb{freescales}
\en
Then $0<\delta_j<1$ and $\sum_j\delta_j<\infty$.
Comparison of even spectral powers at the old temperature gives
\eq
W_{2n_j,D}(\beta_j)\le\lambda_j^{n_j}W_{n_j,D}(\beta_j).
\lb{freeolddouble}
\en
At the new temperature, retain the leading contribution and use
\rlb{freerelative}, with $\beta_{j+1}=2\beta_j$, to get
\eq
\begin{split}
W_{2n_j,D}(\beta_{j+1})
 &\ge\lambda_{j+1}^{2n_j}c_{j+1}^{\free}\\
 &\ge\frac{\lambda_{j+1}^{n_j}}{1+\delta_j}
                                      W_{n_j,D}(\beta_{j+1}).
\end{split}
\lb{freenewdouble}
\en
Dividing \rlb{freenewdouble} by the square of
\rlb{freeolddouble}, and then applying \rlb{puritymonotone}, yields
\eq
\begin{split}
\tau_D^{\free}(n_{j+1},\beta_{j+1})
 &\ge\tau_D^{\free}(2n_j,\beta_j)\\
 &\ge\tau_D^{\free}(n_j,\beta_j)
           \frac{r_j^{n_j}}{1+\delta_j}\\
 &\ge\tau_D^{\free}(n_j,\beta_j)
 \frac{\{(1-u_j)(1-u_{j+1})\}^{R/2}}{1+\delta_j}.
\end{split}
\lb{freeiteration}
\en
The last line uses \rlb{rbound} and $n_j=Rk_j$.
Every initial density matrix on $3^{n_0}$ dimensions has purity at
least $3^{-n_0}$. Therefore
\eq
\tau_D^{\free}(n_j,\beta_j)\ge \tau_*:=3^{-n_0}
 \prod_{\ell=0}^{\infty}
 \frac{\{(1-u_\ell)(1-u_{\ell+1})\}^{R/2}}{1+\delta_\ell}>0.
\lb{freepositivepurity}
\en
The product is positive because both $\sum_\ell u_\ell$ and
$\sum_\ell\delta_\ell$ converge, by \rlb{errorscales} and
\rlb{freescales}. Thus $\tau_*$ is a strictly positive
constant independent of $D$. Each factor in the product lies in
$(0,1)$, so $\tau_*\le3^{-n_0}<1$. Its numerical size is
not needed, and no additional finite-chain numerical input is required.

We next extend this lower bound to lengths of order $k_j^2$,
including odd lengths. If $N\ge n_j$, comparison of the spectral
powers at lengths $N$ and $n_j$ in \rlb{freetransfer} gives
\eq
W_{N,D}(\beta_j)\le\lambda_j^{N-n_j}W_{n_j,D}(\beta_j).
\lb{freeoldlong}
\en
For odd $N$, bound each signed spectral contribution by its absolute
value before comparing powers. At the new temperature,
\rlb{freerelative} at $N$ and $n_j$ gives
\eq
W_{N,D}(\beta_{j+1})\ge
 \lambda_{j+1}^{N-n_j}W_{n_j,D}(\beta_{j+1})
                       \frac{1-\delta_j}{1+\delta_j}.
\lb{freenewlong}
\en
It follows from \rlb{freepositivepurity}, \rlb{freeoldlong},
and \rlb{freenewlong} that
\eq
\tau_D^{\free}(N,\beta_j)\ge
 \tau_* r_j^{N-n_j}\frac{1-\delta_j}{1+\delta_j}.
\lb{freelongpurity}
\en
Choose $j=j(N)$ by \rlb{quadraticlengths}. For large $j$,
$N\ge n_j$. Moreover,
\eq
1\ge r_j^{N-n_j}\ge r_j^N
 \ge\{(1-u_j)(1-u_{j+1})\}^{N/(2k_j)}\longrightarrow1,
\lb{freerlong}
\en
uniformly for $k_j^2\le N<4k_j^2$: the exponent is at most
$2k_j$, while $k_ju_j\longrightarrow0$. Together with $\delta_j\to0$ from \rlb{freescales}, this
shows that, at all sufficiently large lengths of either parity,
\eq
\tau_D^{\free}(N,\beta_{j(N)})\ge \tau_*/2.
\lb{freepurityfinal}
\en
The iteration has preserved positive physical purity, not forced it
toward one. This is exactly the information needed in the final step.

\subsection{Local Gibbs convergence and arbitrary ground vectors}

Let $\rho_{N,\beta,D}^{\free}$ denote the free-end Gibbs expectation.
Its insertion formula is \rlb{fieldGibbsraw} with both endpoint
vectors replaced by $v_\beta^{\free}$. For a local operator $A$,
we use $m,d_{\rm L},d_{\rm R}$ as defined in \rlb{distances}. Equations
\rlb{operatorconcentration} and \rlb{freeoverlap} imply
\eq
\frac{\norm{K_{\beta_j,D}^{d}v_{\beta_j}^{\free}
                   -\lambda_j^dP_jv_{\beta_j}^{\free}}}
     {\lambda_j^d\sqrt{c_j^{\free}}}
 \le e^{\beta_j}q^d.
\lb{freevectorerror}
\en
Expanding both propagated endpoint vectors around their leading
components, and using \rlb{freerelative} for the denominator and
\rlb{Bnormalizednorm} for the insertion, gives, when
$e^{2\beta_j}q^N<1$,
\eq
\begin{split}
&|\rho_{N,\beta_j,D}^{\free}(A)-\nu_{\beta_j,D}(A)|\\
&\quad\le\frac{\norm A}{1-e^{2\beta_j}q^N}
 \left[
 e^{(3+D_0)m\beta_j}
 \{e^{\beta_j}(q^{d_{\rm L}}+q^{d_{\rm R}})
          +e^{2\beta_j}q^{d_{\rm L}+d_{\rm R}}\}
       +e^{2\beta_j}q^N\right].
\end{split}
\lb{freeGibbserror}
\en
The derivation is the same as that of \rlb{exactfielderror}, now
using the overlap bound \rlb{freeoverlap}. For $j=j(N)$ chosen
in \rlb{quadraticlengths} and a fixed support of $A$, the right
side tends to zero as $N\uparrow\infty$. Together with \rlb{nulimit},
this proves
\eq
\rho_{N,\beta_{j(N)},D}^{\free}(A)\longrightarrow\omega_D(A).
\lb{freeGibbslimit}
\en

Write $\hat\rho_{N,\beta,D}^{\free}$ for the corresponding density matrix,
using the same hat convention as in \rlb{Gibbsdensitymatrix}.
Its maximal eigenvalue $p_0$ is the Gibbs probability of each
individual ground-state eigenvector. Since
$\Tr(\hat\rho^2)\le p_0$, \rlb{freepurityfinal} gives
\eq
\hat\rho_{N,\beta_{j(N)},D}^{\free}
 \ge \frac{\tau_*}2\,|\psi\rangle\langle\psi|,
\lb{freegrounddomination}
\en
for every normalized ground vector $\psi$ and all sufficiently
large $N$. The inequality remains
valid when the ground eigenspace is degenerate: the Gibbs density
matrix acts there as $p_0$ times the identity.
For any choice of such ground vectors, write their expectations as
$\omega_{N,D}^{\free}$. Since $0<\tau_*/2<1$, equation
\rlb{freegrounddomination} gives a state $\sigma_{N,D}$ with
\eq
\rho_{N,\beta_{j(N)},D}^{\free}
 =\frac{\tau_*}2\,\omega_{N,D}^{\free}
  +\left(1-\frac{\tau_*}2\right)\sigma_{N,D}.
\lb{freeconvexdecomposition}
\en
Extend these finite-chain states outside $\Lambda_N$ in any way,
and take simultaneous local subsequential limits. By
\rlb{freeGibbslimit}, the limiting relation is
$\omega_D=(\tau_*/2)\omega'+(1-\tau_*/2)\sigma$.
The state $\omega_D$ is pure, hence extremal, so
$\omega'=\omega_D$; see also \cite[Theorem~2.3.19]{BR1}.
Every sequence of ground-vector states therefore has the same
local subsequential limit. If the supremum in \rlb{freelimit} failed to vanish, a violating
sequence of ground vectors would contradict this conclusion.
This proves Theorem~\ref{T:free}.

The argument neither counts the low-lying free-end states nor bounds
a gap above them. It proves a local limit even when the finite-chain Gibbs density matrix need not approach a
rank-one projection. The use of extremality in
the last step is qualitative and does not supply the rate proved for
periodic and endpoint-field chains.

\section{Discussion}
\label{S:scope}

We have proved uniqueness of the translation-invariant ground state
of the infinite $S=1$ Heisenberg antiferromagnet and its stability
under sufficiently small single-ion anisotropy. The state is pure,
locally unique and gapped, and its local expectations are continuous
in $D$. Ground states of the even periodic chains, odd open chains
with equal endpoint fields, and open chains without endpoint fields
all converge locally to this state. In the last case the conclusion
holds for every choice of ground vector, without a gap assumption
for the finite chains.

The development of these results illustrates several uses of the
concentration estimates developed in the OpenAI papers
\cite{OpenAIperiodic,OpenAIboundary} beyond the finite-chain gap
problem.
The local-insertion representation first allows a comparison of
expectation values under different boundary conditions. For open
chains without endpoint fields, a positive lower bound on Gibbs
purity is sufficient: extremality of the infinite-chain limiting
state replaces convergence of the finite-chain density matrix to
a single ground-state projection. Finally, the comparison of
leading transfer eigenvalues yields the free-energy estimate of
Lemma~\ref{L:freeenergy}. Combined with separated insertions and
the minimum-energy characterization, this estimate identifies an
arbitrary translation-invariant ground state, rather than only a
state obtained from a specified sequence of finite chains.

Theorem~\ref{T:main} identifies all translation-invariant algebraic
ground states, not only states produced by a selected finite-chain
construction. Purity gives a further immediate consequence. If a
ground state $\theta$ is invariant under translation by a positive
integer $p$, then
$p^{-1}\sum_{r=0}^{p-1}\theta\circ\tau_r$ is a one-site
translation-invariant ground state and is therefore $\omega_D$.
Since $\omega_D$ is pure, every term in this convex combination
must equal it. Thus no distinct ground state with any finite
translation period exists.

The periodic, equal-field, and free-end constructions all select
$\omega_D$. These assertions do not prove that every infinite-chain
ground state of $H(D)$, in the standard operator-algebraic sense,
is obtained from one of these constructions or is translation
invariant. Although the present method seems likely to extend to a
much broader class of boundary conditions, establishing uniqueness
among all infinite-chain ground states requires an additional
argument. No claim about arbitrary
nontranslation-invariant ground states is made here.

It is useful to return to the AKLT precedent discussed in the
introduction. The original argument identifies the states with
zero expectation of every positive local interaction, and hence
all translation-invariant ground states \cite{AKLT1988}.
The additional difficulty is to exclude states that satisfy the
operator-algebraic local stability condition without minimizing
each local interaction separately. Matsui resolved this problem
for the AKLT chain and proved uniqueness without assuming
translation invariance \cite{Matsui1997}. His argument uses the
structure of zero-energy states on the two half-infinite chains,
including quasi-equivalence of the zero-energy states on each
half-chain. Thus even in the frustration-free AKLT model,
uniqueness among all algebraic ground states required more than
the exact finite-chain ground-state description. For the Heisenberg
and $J$-$D$ chains, the local interactions cannot be minimized
simultaneously, and the corresponding exclusion of all
nontranslation-invariant ground states remains beyond the present
proof.

\appendix
\section{Verification of the finite criteria}
\label{S:numerics}

This appendix verifies the periodic criterion \rlb{periodiccriterion}
and the additional boundary criterion \rlb{boundarycriterion}
for one simultaneous choice of constants. All terminating decimals below denote exact
rational numbers. The original thermal and variational certificates are
used as inputs; their underlying matrix computations are not repeated.

\subsection{Input at the Heisenberg point}

Choose
\eq
h_0=\frac35,\qquad \beta_0=392,\qquad k_0=960,
\qquad \eta=\frac1{100}.
\lb{numericalchoices}
\en
Thus $2k_0=1920$. In this subsection $\lambda_0$, $c_0$, and
$e_{k_0}(0)$ are evaluated at $D=0$. The estimates imported from
\cite{OpenAIboundary} are
\eqa
S_0(k_0,\beta_0)&\ge1-0.000015,&
T_0(2k_0,\beta_0)&\ge1-0.000015,
\lb{seedperiodic0}\\
e_{k_0}(0)&\le0.001,&
\tau_0(2k_0,\beta_0)&\ge1-0.001.
\lb{seedboundary0}
\ena
Here $e_{k_0}(0)$ is the boundary defect at $D=0$ and
$\beta=\beta_0$, with the definition in \rlb{boundaryquantities}.
The first line follows from Proposition~3.4 of
\cite{OpenAIboundary}, at its level $j=5$, where its periodic defect
is strictly less than $0.000015$.
The derivation of the first line in that proposition uses only the
finite periodic inputs of \cite{OpenAIperiodic}; no endpoint-field
thermal or variational certificate is involved. Thus Theorems~\ref{T:main}, \ref{T:periodic}, and
\ref{T:free} require only periodic numerical inputs,
even though this convenient improved periodic scale is recorded in
\cite{OpenAIboundary}. The second line, which is needed only for
Theorem~\ref{T:boundary}, is the level-$5$ conclusion in its
table~(5.9).
That manuscript obtains the estimates from its finite thermal and
variational certificates and the periodic input of
\cite{OpenAIperiodic}.
Our iteration starts at this already improved scale.
All ratios and normalized spectral quantities used here are unchanged
by the scalar energy normalization of those manuscripts.

To convert the imported overlap estimate into the fourth physical ratio,
write
\eq
\mathcal C_0(k_0,\beta_0)
 =\frac{\lambda_0^{k_0}}{Z_{k_0,0}(\beta_0)}
      \frac{Y_{2k_0,0}(\beta_0)}
           {\lambda_0^{k_0}Y_{k_0,0}(\beta_0)}
 \ge(1-0.000015)(1-0.001)=0.998985015.
\lb{C0bound}
\en
For the first factor, the comparison of even eigenvalue powers used in
\rlb{maximumeigenvalue}, now at the initial temperature, gives
\eq
\frac{\lambda_0^{k_0}}{Z_{k_0,0}(\beta_0)}
 \ge\frac{Z_{2k_0,0}(\beta_0)}{Z_{k_0,0}(\beta_0)^2}
 =S_0(k_0,\beta_0).
\lb{seedleadingratio}
\en
For the second factor, retain the leading contribution to obtain
\eq
\frac{Y_{2k_0,0}(\beta_0)}{\lambda_0^{k_0}Y_{k_0,0}(\beta_0)}
 \ge\frac{\lambda_0^{k_0}c_0}{Y_{k_0,0}(\beta_0)}
 =1-e_{k_0}(0).
\lb{seedoverlapratio}
\en
Thus all four inequalities of \rlb{finitecriterion} have a strict
margin at $D=0$.

\subsection{Stability of the finite criterion in \texorpdfstring{$D$}{D}}

All ratios in \rlb{finitecriterion} are continuous functions of $D$ at
fixed $k_0,\beta_0,h_0$. Strict inequalities at $D=0$ therefore give the
criterion on some interval $[-D_0,D_0]$.
For an explicit interval we use the following simple bounds, which also
explain why no new transfer-spectrum calculation at $D\ne0$ is required.

For $b\in\{\per,\fld,\opp\}$,
$H_N^b(D)=H_N^b(0)+D\mathcal Q_N$, where
$0\le\mathcal Q_N:=\sum_{j\in\Lambda_N}Q_j\le N\one$.
If $\widetilde Z_{N,D}(\beta)$ denotes its partition function, min--max comparison
of the eigenvalues gives
\eq
\begin{array}{ll}
e^{-\beta ND}\widetilde Z_{N,0}(\beta)\le \widetilde Z_{N,D}(\beta)\le \widetilde Z_{N,0}(\beta),
 & D\ge0,\\[2mm]
\widetilde Z_{N,0}(\beta)\le \widetilde Z_{N,D}(\beta)
 \le e^{\beta N|D|}\widetilde Z_{N,0}(\beta),
 & D\le0.
\end{array}
\lb{traceperturbation}
\en
These are trace inequalities obtained from eigenvalue bounds; they do not
use operator monotonicity of exponentiation.
At length $2k$ and for either sign of $D$, they imply
\eqa
S_D(k,\beta)&\ge e^{-2\beta k|D|}S_0(k,\beta),&
\mathcal C_D(k,\beta)&\ge e^{-2\beta k|D|}\mathcal C_0(k,\beta),
 \lb{perturbS}\\
T_D(2k,\beta)&\ge e^{-4\beta k|D|}T_0(2k,\beta),&
\tau_D(2k,\beta)&\ge e^{-4\beta k|D|}\tau_0(2k,\beta).
 \lb{perturbT}
\ena
These estimates reduce the perturbation in $D$ to a comparison at
fixed finite lengths and inverse temperatures. No new computation of
the transfer spectrum at $D\ne0$ is required.
We now apply these bounds to the strict margins just established.

\subsection{An explicit interval and gap}

Take
\eq
D_0=\frac1{200000000}=5\times10^{-9}.
\lb{D0value}
\en
The trace bounds \rlb{perturbS}--\rlb{perturbT} apply for both signs
of $D$. With
\eq
t_0=2\beta_0k_0D_0=\frac{294}{78125}=0.0037632,
\lb{tseed}
\en
use $e^{-t_0}\ge1-t_0$ and $e^{-2t_0}\ge1-2t_0$.
The resulting lower bounds, valid uniformly for $|D|\le D_0$, are
\begin{center}
\begin{tabular}{c|c|c}
Quantity & Lower bound used & Exact result\\ \hline
$S_D(k_0,\beta_0)$
 &$(1-t_0)(1-0.000015)$&$0.996221856448$\\
$T_D(2k_0,\beta_0)$
 &$(1-2t_0)(1-0.000015)$&$0.992458712896$\\
$\tau_D(2k_0,\beta_0)$
 &$(1-2t_0)(1-0.001)$&$0.9914811264$\\
$\mathcal C_D(k_0,\beta_0)$
 &$(1-t_0)(0.998985015)$&$0.995225634591552$
\end{tabular}
\end{center}
The first three entries exceed $1-\eta=0.99$, and the last exceeds
$1-\eta/2=0.995$.
This verifies \rlb{finitecriterion}. The periodic inequalities complete
Theorem~\ref{T:main} by Section~\ref{S:uniqueness},
Theorem~\ref{T:periodic} by Sections~\ref{S:bulkconsequences}
and \ref{S:quantitativerate}, and Theorem~\ref{T:free} by
Section~\ref{S:free}. The additional two boundary inequalities
complete Theorem~\ref{T:boundary} by
Proposition~\ref{P:finiteapproximation} and
Section~\ref{S:limit}.
The common gap lower bound is
\eq
\gamma_0=\frac{\log10}{392}\simeq0.005874.
\lb{numericalgap}
\en
The finite thresholds supplied by Proposition~\ref{P:finiteapproximation}
are even
periodic $N\ge960$ and odd field-chain $N\ge1921$, that is,
$L\ge480$ for the periodic gap in Theorem~\ref{T:periodic}
and $L\ge960$ for the gap in Theorem~\ref{T:boundary}.
The value of $D_0$ is a conservative consequence of the finite trace
bounds; it is not intended to describe the maximal parameter interval
on which translation-invariant uniqueness and the boundary conclusions hold.

\section*{Acknowledgment}
\addcontentsline{toc}{section}{Acknowledgment}
I thank OpenAI for making the papers \cite{OpenAIperiodic,OpenAIboundary}
available. Their proofs of the periodic and boundary-field Haldane-gap
statements also open new possibilities for the mathematical study of
quantum spin chains.
This work is supported by JSPS Grants-in-Aid for Scientific Research
No.\ 25K07171.

\end{document}